\documentclass{sigplanconf}

\pdfoutput=1

\usepackage{amsmath,amssymb,censor,cmll,flushend,mathwidth,mathtools,multicol,stmaryrd,xspace}
\usepackage[dvipsnames]{xcolor}
\usepackage[hyphens]{url}

\usepackage{infer}
\newcommand{\isp}{\hspace{\infskip}}

\usepackage{float}
\floatstyle{boxed}
\restylefloat{figure}

\usepackage{zi4}
\DeclareMathAlphabet{\mathtt}{OT1}{zi4}{m}{n}

\usepackage{amsthm}
\newtheorem{theorem}{Theorem}
\newtheorem{lemma}[theorem]{Lemma}

\theoremstyle{definition}
\newtheorem{defn}[theorem]{Definition}

\newcommand{\secref}[1]{(\S\ref{sec:#1})}

\newcommand{\lemref}[1]{Lemma~\ref{thm:#1}}
\newcommand{\thmref}[1]{Theorem~\ref{thm:#1}}
\newcommand{\figref}[1]{Figure~\ref{fig:#1}}

\newcommand{\mkwd}[1]{\ensuremath{\text{\texttt{\underbar{#1}}}}}
\newcommand{\key}[1]{\mkwd{#1}}
\newcommand{\mlet}[3]{\mkwd{let}\;#1 = #2\;\mkwd{in}\;#3}

\newcommand{\inl}[1]{\mathsf{in_1}\,#1}
\newcommand{\inr}[1]{\mathsf{in_2}\,#1}
\newcommand{\ini}[1]{\mathsf{in}_i\,#1}

\newcommand{\inii}[2]{\mathsf{in}_i^{#1}\,#2}
\newcommand{\case}[5]{\mkwd{case}\,#1\,\mkwd{of}\;\{\inl{#2} \mapsto #3; \;\inr{#4} \mapsto #5\}}
\newcommand{\then}{\ensuremath{\Rightarrow}}
\newcommand{\labto}[1]{\ensuremath{\stackrel{#1}{\to}}}
\newcommand{\lto}{\labto{\circ}}
\newcommand{\uto}{\labto{\bullet}}
\newcommand{\pair}[2]{\ensuremath{#1 \otimes #2}}
\newcommand{\opt}[2]{\ensuremath{#1 \oplus #2}}

\newcommand{\lami}[3]{\ensuremath{\lambda^{#1}#2.#3}}

\newcommand{\predh}[2]{\ensuremath{#1\ #2}}

\newcommand{\Un}[1]{\predh{\mathtt{Un}}{#1}}
\newcommand{\Jun}[1]{#1\,\mathsf{un}}

\newcommand{\Unl}[1]{\Un{#1}}

\newcommand{\Junl}[1]{\Jun{#1}}
\newcommand{\Fun}[1]{\predh{\mathsf{Fun}}{#1}}

\newcommand{\moreunlimited}{\geq}
\newcommand{\lessgeneral}{\sqsubseteq}
\newcommand{\trule}[1]{{\textsc{(#1)}}\xspace}
\newcommand{\strule}[1]{{\trule{#1$^{\mathrm{S}}$}}\xspace}

\newcommand{\C}{ctr}
\newcommand{\W}{wkn}

\newcommand{\I}[1]{\ensuremath{#1\!}~I}
\newcommand{\E}[1]{\ensuremath{#1\!}~E}
\newcommand{\vdashsup}[1]{\ensuremath{\vdash^{\raisebox{1px}{\tiny\!\!\!#1}}}}
\newcommand{\vdashS}{\vdashsup{{\upshape S}}}

\newcommand{\entails}{\ensuremath{\then}}
\newcommand{\nentails}{\ensuremath{\not\then}}
\newcommand{\restrict}[2]{#1|_{#2}}
\newcommand{\Eta}{\ensuremath{H}}
\newcommand{\Set}[1]{\{ #1 \}}
\newcommand{\Bag}[1]{\lbag #1 \rbag}
\newcommand{\subv}[2]{[\vec{#1}/\vec{#2}]}

\newcommand{\dom}{\mathsf{dom}}

\newcommand{\red}[4]{\ensuremath{#1 \Downarrow^{#3}_{#4} #2}}
\newcommand{\seq}{\vec}

\newcommand{\lang}{Quill\xspace}

\newcommand{\theIH}{the induction hypothesis\xspace}

\newcommand{\Fpop}{F$^\circ$\xspace}

\newcommand{\gvout}[2]{#1 \mathbin{\mathtt{!}} #2}
\newcommand{\gvin}[2]{#1 \mathbin{\mathtt{?}} #2}
\newcommand{\gvend}{\ensuremath{\mathtt{End}}}

\newenvironment{syntax}%
{\[\begin{array}{lr@{\hspace{5px}}r@{\hspace{5px}}l}}
{\end{array}\]\ignorespacesafterend}

\newenvironment{longsyntax}%
{\[\begin{array}{lrr@{\hspace{2px}}r@{\hspace{2px}}l}}
{\end{array}\]\ignorespacesafterend}

\newcommand{\tcl}[1]{\multicolumn{2}{l}{#1}}
\newcommand{\tcr}[1]{\multicolumn{2}{r}{#1}}

\newcommand{\bs}{\char`\\}

\usepackage{listings}
\lstnewenvironment{code}{}{}
\lstnewenvironment{codef}{\lstset{basicstyle=\ttfamily\small,xleftmargin=1.5em}}{}

\newcommand{\thetitle}{The Best of Both Worlds}
\newcommand{\thesubtitle}{Linear Functional Programming without Compromise}

\usepackage[pdftex,
            pdfauthor={J. Garrett Morris},
            pdftitle={\thetitle: \thesubtitle}]{hyperref}
\hypersetup{colorlinks=true,allcolors=black}

\title\thetitle
\subtitle\thesubtitle
\authorinfo{J. Garrett Morris}
           {The University of Edinburgh, UK}
           {Garrett.Morris@ed.ac.uk}
\toappear{}

\begin{document}

\maketitle

\begin{abstract}
  We present a linear functional calculus with both the safety guarantees expressible with linear
  types and the rich language of combinators and composition provided by functional programming.
  Unlike previous combinations of linear typing and functional programming, we compromise neither
  the linear side (for example, our linear values are first-class citizens of the language) nor the
  functional side (for example, we do not require duplicate definitions of compositions for linear
  and unrestricted functions).
  To do so, we must generalize abstraction and application to encompass both linear and unrestricted
  functions.  We capture the typing of the generalized constructs with a novel use of qualified
  types.
  Our system maintains the metatheoretic properties of the theory of qualified types, including
  principal types and decidable type inference.
  Finally, we give a formal basis for our claims of expressiveness, by showing that evaluation
  respects linearity, and that our language is a conservative extension of existing functional
  calculi.
\end{abstract}



\section{Introduction}

Integers have a pleasing consistency: values do not become more or less integers over the course of
a computation.  The same is not true for file handles: we can no longer expect to read from or write
to a file handle after it has been closed.  Traditional functional type systems, like the logic they
resemble, are good for integers (i.e., unchanging propositions), but less good for file handles
(i.e., temporary ones).  If our type systems are to help in the latter case, we need ones with a
different logical character.

One approach is suggested by Girard's linear logic~\cite{Girard87}, which requires that each
hypothesis be used exactly once in the course of a proof.  Intuitively, linear propositions are
finite resources, which can neither be duplicated nor discarded, rather than arbitrary truth values,
available as often, or as rarely, as needed.  Linear type systems adopt the same approach to
variables: each bound variable must be used exactly once in the body of its binder.  Such type
systems have been used to reason about resource usage and concurrency.  For example, they have been
used to assure safe manipulation of state~\cite{GuzmanH90,AhmedFM05}, regulate access to shared
resources~\cite{FahndrichD02,Boyland03}, and type interacting concurrent
processes~\cite{VasconcelosGR06,CairesPfenning10,Wadler14}.  Each of these examples uses the
restrictions on reuse and discard introduced by linearity to assure safety invariants.
Simultaneously, several general purpose linear functional languages have been proposed, including
those of Wadler~\cite{Wadler93}, Mazurak et al.~\cite{MazurakZZ10}, and Tov and
Pucella~\cite{TovP11}.  However, attempts to adapt functional programming idioms and abstractions to
these calculi are thwarted by the interplay of linear and intuitionistic types.  This paper proposes
a novel combination of linear and qualified types that provides the safety of linear types without
losing the expressiveness of functional programming.

We identify three requirements for the integration of linear types and functional programming.
As an example, consider the $K$ combinator, defined as $\lambda x. \lambda y. x$.
We begin with its arguments: argument $x$ is used once, and so can take on values of any type.
Argument $y$ is discarded, and so can only take on values of unrestricted type.  This illustrates
the first requirement: we must distinguish between quantification over arbitrary type and quantification
over unrestricted types.  Next, consider the application $K\,V$, giving a new function
$\lambda y. V$.  Whether we can reuse this function depends on $V$.  If $V$ is a Boolean or integer
value, for example, there is no danger in reusing $\lambda y. V$.  On the other hand, if $V$ is a
file handle or capability, then reuse of the function would imply reuse of $V$, and should be
prevented.  This illustrates the second requirement: we must distinguish between linear and
unrestricted functions, which distinction is determined by the environment captured by each
function.  In particular, there is no single static characterization of the linearity of the
subterms of $K$ that accounts for its application to both linear and unrestricted values.  Finally,
consider the composition function $\circ$, defined as $\lambda f. \lambda g. \lambda x. f\,(g\,x)$.
We know that we must be able to apply $f$ and $g$ to things, and that both are used linearly,
However, these constraints are satisfied by both linear and unrestricted functions.  So, the final
requirement is that we must generalize the typing of application to range over the possible types of
function.

These requirements have been addressed in previous work, although the interactions between them have
not.  Quantification over unrestricted values can be expressed in a linear system using either kinds
(and subkinding)~\cite{MazurakZZ10} or type classes~\cite{GanTM14}; both approaches extend naturally
to account for pairs and sums.  Many existing systems use subtyping to account for application,
either implicitly~\cite{MazurakZZ10,GanTM14} or explicitly~\cite{GayV10,TovP11}.  Tov and
Pucella~\cite{TovP11} introduce a notion of relative arrow qualifiers, sufficient to express the
typing of $K$, but at the cost of significant complexity in their type system.  The interplay of
these disparate mechanisms has not been fully explored.  For example, none the existing systems can
express the desired typing of composition, nor have they been shown to support complete type
inference.

We propose a new, uniform approach to integrating linear types and functional programming, based on
the theory of qualified types~\cite{Jones94}.  Rather than invent new type system features, we
present a language design based on a novel combination of qualified and linear typing, both
existing, well-studied type systems.  To demonstrate our design, we return to the $K$ combinator, to
which we give the (qualified) type $(\Unl u,t \moreunlimited f) \then t \to u \labto f t$.  First,
we observed that the second argument (here typed by $u$) must be unrestricted; this is captured by
the predicate $\Unl u$.  Second, we observed that the result of $K\,V$, $\lambda y. V$, may be
either an linear or unrestricted function, depending on the linearity of $V$.  We capture this
through the use of two predicates, one that identifies functions and another that specifies relative
linearity.  The predicate $\Fun f$ is satisfied only when $f$ is a function type; we write
$t \labto f u$ to denote the type $f\,t\,u$ under the predicate $\Fun f$.  Here, we use it to range
over the possible types for $\lambda y.V$; we would make a similar use of the $\Fun\!$ predicate to
express the argument types of the composition operator.  The predicate $t \geq f$ is satisfied when
$t$ supports more structural rules (i.e., duplication and discard) than $f$.  Thus, in the typing of
the $K$ combinator, if $t$ is linear, then $f$ must be linear; alternatively, if $t$ is
unrestricted, than $f$ can be either linear or unrestricted.

Formally, we capture our approach in the design of a core linear calculus, which we call a Qualified
Linear Language (\lang).  \lang is a linear variant of Jones's calculus OML, extended with
Haskell-like first-class polymorphism~\cite{Jones97}, and with entailment rules for the
$\mathsf{Un}$, $\mathsf{Fun}$, and $\moreunlimited$ predicates.  We preserve the metatheoretic
properties of OML, particularly principal types and decidable type inference, without requiring the
programmer to provide type or linearity annotations.  We show that our system is a conservative
extension of (non-linear) OML; concretely, this means that we can view our approach as giving linear
refinements of existing functional languages and idioms, rather than replacing them entirely.
Finally, we give a natural (big-step) semantics for \lang and show that evaluation respects
linearity.

In summary, this paper contributes:
\begin{itemize}
\item The design and motivation of \lang, including examples of \lang's application to prototypical
  uses of linear types (dyadic session types) and higher-order functional programming
  (monads)~\secref{programming}.
\item A formal account of the \lang type system and its relationship to OML, including a sound and
  complete type inference algorithm~\secref{typesystem}.
\item A linearity-aware semantics of \lang, and a proof that values of linear type are neither
  duplicated nor discarded during evaluation~\secref{semantics}.
\item A discussion of further extensions of \lang, including its applicability to other
  substructural type systems, such as affine or relevant typing~\secref{extensions}.
\end{itemize}
We begin with an introduction to linear type systems and their uses~\secref{background}, and
conclude by discussing related~\secref{related} and future~\secref{future} work.

\section{Substructural Type Systems}\label{sec:background}


Before describing the details of our language, we give examples of several applications of
substructural type systems and several general-purpose substructural calculi.

\subsection{Applications of Substructural Typing}


Linear type systems restrict the use of weakening (i.e., discarding variables) and contraction
(i.e., reusing variables), allowing us to reason about state and resource usage in programs.  For
example, excluding weakening could prevent memory or resource leaks, by requiring that each input to
a computation be consumed during its evaluation.  Excluding contraction could be used in describing
component layouts in circuits, where a limited number of each computational unit are available.
Linear type systems combine these, providing exact control over resource usage.
This section describes two uses of linear types: session types and referentially-transparent
in-place update.  These demonstrate two different uses of linear types: session types evolve over
the course of a computation, capturing changes in underlying state, while mutable values must be
used linearly to preserve referential transparency.  In each case, we make two points: first, the
need for the restrictions introduced by linearity, and second, the difficulties introduced in
integrating the use of linear and unrestricted types.

\paragraph{Mutable arrays.}

We begin by considering in-place update.  Suppose that we want to be able to read from and update
arrays in a referentially transparent way.  We might expect each update to produce a new copy of the
array; otherwise, updates would be visible through other references to the original array.  For
large arrays, this copying will be extremely costly, both in time and space.  The copying could be
avoided if we could ensure that the use of arrays was single-threaded.  That is, so long as no
``old'' copies of arrays are ever used, updates can be performed in place. Chen and
Hudak~\cite{ChenH97} consider the connection between single-threaded usage, potentially enforced
monadically, and linearity.  They introduce an affine type system (they allow discarding but not
duplication), and show that updating linearly-typed values can safely be performed in place.  They
also show how the operations on a linear data type can be interpreted to give a monad in an
intuitionistic calculus, while preserving the safety of in-place update.  However, this approach
relies on hiding the linearly typed values, making them second-class citizens of the non-linear
calculus.  That is, while the monadic computations describe transformations of an underlying array,
they give no direct access to the array itself.  Consequently, while their approach applies to
linear arrays of unrestricted element types, it could not apply to linear element types (such as
other arrays), because it gives direct access to the array elements.

\paragraph{Session types.}

Next we consider session types, an instance of behavioral typing.  Communication protocols
frequently specify not just what data can be sent, but in what order.  For example, the Simple Mail
Transfer Protocol specifies not just a list of commands (identifying senders, recipients, message
bodies, and so forth), but also a particular ordering to messages (the sender's address must precede
the recipients' addresses, which must precede the message body).  Session types, originally proposed
by Honda~\cite{Honda93}, provide a mechanism for capturing such expectations in the types of
communication channels.  The critical aspect of his type system is that types evolve over the course
of a computation to reflect the communication that has already taken place.
For example, if channel $c$ has session type $\gvout{Int}{\gvin{Int}\gvend}$, we expect to send an
integer along $c$, then receive an integer from $c$.  After we have sent an integer, the type of $c$
must change to $\gvin{Int}{\gvend}$, reflecting the remaining expected behavior.  We can implement
session types in a functional setting by giving channels linear types, and reflecting the evolution
of types in the type signatures of the communication primitives:
\begin{align*}
  \mathsf{send} &:: \pair t {(\gvout t s)} \to s &
  \mathsf{receive} &:: (\gvin t s) \to \pair{t}{s}
\end{align*}
Continuing the example above, we see that the result of $\mathsf{send} \, (4, c)$ will be of type
$\gvin{Int}\gvend$, as we hoped.  The linearity of these channels is crucial to assuring the type
correctness of communication: reusing channel $c$ would allow us to send arbitrarily many integers,
not just one.  There are approaches to encoding session types in existing functional languages, such
as that of Pucella and Tov~\cite{PucellaT08}, but they result in channels being second class values.
For example, sending or receiving channels requires different primitives from those for sending or
receiving other values.

\subsection{General-Purpose Linear Calculi}\label{sec:other-linear}

Wadler~\cite{Wadler93} gives a $\lambda$-calculus based on Girard's logic of unity, a refinement of
linear logic.  In his approach, the types (ranged over by $\tau,\upsilon$) are precisely the
propositions of linear logic, including pairs ($\tau \otimes \upsilon$), functions
($\tau \multimap \upsilon$), and the exponential modality ($!\tau$).  His type system tracks two
kinds of assumptions, linear ($x: \langle \tau \rangle$) and intuitionistic ($x: [\tau]$); only the
latter are subject to contraction and weakening.  Wadler does not include polymorphism in his
calculus; nevertheless, we can see that his treatment of intuitionistic types would preclude
attempts toward generality.  He gives explicit term constructors to introduce and eliminate the
exponential modality, and these constructs surround any use of unrestricted types.  If $M$ is of
type $\tau \multimap \upsilon$, and $N$ is of type $\tau$, then we can construct the application
$M\,N$ of type $\upsilon$; on the other hand, if $M$ is of type $!(\tau \multimap \upsilon)$, then
we must explicitly eliminate the $!$ constructor at each use of $M$, as $\mlet{!f}{M}{f\,N}$.
Returning to our introductory example, we have two families of types (and corresponding terms) for
the $K$ combinator, $!(!\tau \multimap {!(!\upsilon \multimap {!\tau})})$ if the first argument is
intuitionistic, and $!(\tau \multimap {!\upsilon} \multimap \tau)$ otherwise.

Mazurak et al.~\cite{MazurakZZ10} present a streamlined, polymorphic linear $\lambda$-calculus.
Their calculus, called \Fpop, extends the Girard-Reynolds polymorphic $\lambda$-calculus with
linearity, and introduces a kind system which distinguishes between linear (kind $\circ$) and
unrestricted (kind $\star$) types.  They then define the kinds of types such as pairs in terms of
the kinds of their components: $\tau \otimes \upsilon$ is of kind $\star$ if both $\tau$ and
$\upsilon$ are of kind $\star$, and must be of kind $\circ$ otherwise.  Finally, they introduce a
subkinding relation, allowing a type of kind $\star$ to be used any place a type of kind $\circ$ is
expected.  This reflects the observation that an unrestricted value can be used any number of times,
including once. While their approach seamlessly encompasses many uses of unrestricted types, it does
not extend to functions.  \Fpop distinguishes between linear functions $\lambda^\circ x.M$, of type
$\tau \lto \upsilon$, which may capture arbitrary variables in their environment, and unrestricted
functions $\lambda^\star x.M$, of type $\tau \uto \upsilon$, which can only capture unrestricted
values.  Consequently, \Fpop still has four distinct types for the $K$ combinator
\begin{gather*}
  \forall (t: \circ). \forall (u: \star). t \uto u \lto t \qquad
  \forall (t: \circ). \forall (u: \star). t \lto u \lto t \\
  \forall (t: \star). \forall (u: \star). t \uto u \uto t \qquad
  \forall (t: \star). \forall (u: \star). t \lto u \uto t
\end{gather*}
each with distinct inhabitants.  The problem is endemic to the use of higher-order functions; for
example, their system has numerous distinct application and composition functions.

Tov and Pucella~\cite{TovP11} present Alms, an affine calculus with a kind system similar to \Fpop
but with additional flexibility in the treatment of functions.  Their treatment of functions
includes not just affine ($\labto{\mathsf{A}}$) and unrestricted ($\labto{\mathsf{U}}$) functions,
but also functions with relative qualifiers.  For example, Alms has a single most-general type for
the $K$ combinator, written
\[
  \forall (t: \mathsf{A}). \forall (u: \mathsf{U}). t \labto{\mathsf{U}} u \labto{t} t.
\]
The arrow $\labto t$ must be more restricted than the instantiation of $t$.  If $t$ is instantiated
to an affine type, then $\labto{t}$ must be $\labto{\mathsf{A}}$; otherwise, it can be
$\labto{\mathsf{U}}$.  They include subtyping explicitly; for example, $Int \labto{\mathsf{U}} Int$
is a subtype of $Int \labto{\mathsf{A}} Int$.  Alms is quite expressive, but this comes at the cost
of high complexity; we believe that \lang provides similar expressiveness in a significantly simpler
setting.


\section{Programming in \lang}\label{sec:programming}

This section gives an intuitive overview of our calculus \lang and its primary features.  We begin
by describing the use of overloading to capture the non-linear use of assumptions.  We then consider
the particular problems arising from having both linear and unrestricted functions, the overloading
of application and abstraction, and introduce the corresponding predicates on types.  Finally, we
consider two examples of programming in \lang: a simple presentation of dyadic session types,
demonstrating the use of linearity, and a Haskell-like presentation of monads, demonstrating the
interaction between linearity and higher-order functional programming.  For the purposes of this
section, we use a Haskell-like syntax for \lang, in which we distinguish linear functions
($\tau \lto \upsilon$) from unrestricted functions ($\tau \uto \upsilon$).  We give a formal account
of \lang's syntax and semantics in the following sections.

\subsection{Contraction and Weakening with Class}\label{sec:class}
\newcommand{\ifte}[3]{\ensuremath{\mkwd{if}\;#1\;\mkwd{then}\;#2\;\mkwd{else}\;#3}}

Our goal is a functional language in which values of some (but not all) types must be treated
linearly.  The central problem is the integration of unrestricted types, and functions on
unrestricted (but otherwise generic) types, with an otherwise linear type system.  We describe one
solution, based on the theory of qualified types.

We begin by distinguishing linear from unrestricted types.  We consider a type to be unrestricted if
values of that type can be duplicated and discarded.  That is, a type $\tau$ is unrestricted if we
can exhibit values of type $\tau \uto 1$ and $\tau \uto \pair \tau \tau$.  (This approach roughly
parallels Filinski's interpretation of intuitionistic types by commutative comonoids in the model of
a linear calculus~\cite{Filinski92}.)  For example, consider a type for Booleans with the standard
branching construct and constants.  We can demonstrate that Booleans are unrestricted by giving the
terms $\lambda b. \ifte{b}{()}{()}$ to discard a Boolean, and
$\lambda b. \ifte{b}{(True, True)}{(False, False)}$ to copy one.  This leaves the problem of how to
write code generic over such types; for instance, we would like the function $\lambda x. (x, x + 1)$
to be applicable to arguments of any unrestricted numeric type.

Our approach is inspired by the use of type classes in Haskell.  Type classes were introduced to
solve similar problems, such as how to write functions generic over types that have an equality
operator, or that can be converted to and from text.  For our purposes, we can imagine introducing a
type class \texttt{Un}, which identifies unrestricted types:
\begin{code}
class Un t where
  drop :: t ->* 1
  dup  :: t ->* t :*: t
\end{code}
The methods of \texttt{Un} provide the defining behavior of an unrestricted type.  We could then
imagine using these methods to implement terms such as the one above, for which we could write
$\lambda x. \mlet{(x,x')}{\mathtt{dup}\,x}{(x,x'+1)}$.  In inferring a type for this term, we would
observe that its argument type has to support numeric operations (and so be a member of the
\texttt{Num}) class, and has to support \texttt{dup} (and so be a member of the \texttt{Un} class).
We would conclude that it should have type $(\texttt{Num}\,t,\Unl t) \then t \uto \pair t t$.

One advantage of this view of unrestricted types is that it extends naturally to products, sums, and
recursive types.  For example, a pair of values $(V, W)$ can  safely be copied only when both $V$
and $W$ could individually be copied.  We can capture this in an instance of the $\texttt{Un}$
class:
\begin{code}
instance (Un t, Un u) => Un (t :*: u) where
  drop (x, y)  =  let () = drop x in drop y
  dup (x, y)   =  ((x', y'), (x'', y'')) where
    (x', x'')  =  dup x
    (y', y'')  =  dup y
\end{code}
The relationship between the linearity of $t$ and $u$ and the linearity of $\pair t u$ arises
organically from the typing of the \texttt{drop} and \texttt{dup} methods.  The argument for sums is
parallel, with the same results.

Of course, we do not intend programmers to use the \texttt{drop} and \texttt{dup} methods directly,
and we imagine that instances of \texttt{Un} would be inferred automatically from type declarations.
Instead, \lang allows variables to be used freely, and infers \texttt{Un} predicates as if any
duplication or discarding of variables had been done explicitly.  Thus, $\lambda x. (x, x+1)$ is a
well-typed \lang term with the type $(\mathtt{Num}\,t,\Unl t) \then t \uto \pair t t$, as above.

\subsection{The Problem of the Copyable Closure}

We have an appealing view of how to distinguish unrestricted from linear types, and how to account
for the linearity of products and sums.  Unfortunately, this view does not extend to provide a
uniform treatment of functions.  Consider the curried pair constructor
$\lambda x. \lambda y. (x,y)$.  We know that the linearity of the resulting pair depends only on the
linearity of its components.  But what about the intermediate result?  Suppose that we apply this
function to some value $V$ giving the term $\lambda y. (V, y)$.  Whether we can copy this term
depends upon the captured value $V$; intuitively, we can say it depends on the function's closure.
However, this is not reflected in the function type. (While the type of $V$ does appear in the
result type, so does the type of $y$, but the linearity of the function type is solely a consequence
of $V$'s type.)  We are thus forced to introduce distinct types for linear and unrestricted
functions.  This section discusses the resulting language design questions: how to handle
application and abstraction in a language with multiple function types, and how to relate the type
of a function to the type of its captured environment.

\subsubsection*{Application}

We begin with application, the simpler of the two problems.  Consider the uncurried application
function.  In intuitionistic calculi, this is $\lambda (f, x). f\,x$, of type $(t \to u, t) \to u$.
In the linear settings, things are not so simple: we must decide whether the argument $f$ and the
function being defined are linear or unrestricted functions.  These choices are independent, giving
four incomparable types:
\begin{gather*}
  \pair {(t \lto u)} t \lto u \qquad \pair {(t \lto u)} t \uto u \\
  \pair {(t \uto u)} t \lto u \qquad \pair {(t \uto u)} t \uto u
\end{gather*}
We can resolve this repetition by observing that (built-in) application is implicitly overloaded: we
would like to write $f\,x$ whether $f$ is a linear or unrestricted function.  We make this
overloading explicit in the types.  We introduce a new predicate, $\Fun f$, which holds when $f$ is
a function type; intuitively, we can think of this as corresponding to a class whose sole method is
application, and whose only members are $\lto$ and $\uto$.  We can then type application with
reference to this class, rather than in terms of either of the concrete function types.  This
reduces the number of application functions from four to two: we have
$\Fun f \then \pair {f\,t\,u} t \lto u$ and $\Fun f \then \pair {f\,t\,u} t \uto u$.  We introduce
syntactic sugar to make the $\mathsf{Fun}$ predicate easier to read.  We will write $t \labto f u$
to indicate the type $f\,t\,u$ constrained by $\Fun f$, and further write $t \to u$ to indicate
$t \labto f u$ for some fresh type variable $f$.  Using this sugar, we arrive at the most general
type for the application function, $\pair{(t \to u)} t \to u$.

In the previous section, we motivated the typing of contraction and weakening using the methods of an
\texttt{Un} class, even though we intend their use to be implicit in practice.  In the case of the
\texttt{Fun} predicate, the class method intuition is less helpful.  Defining primitive application
as a class method is difficult (how would it be used, except by application?), and we will rely on
the \texttt{Fun} predicate holding only for the built-in function types.  This reinforces the
expressiveness of qualified types, even beyond their traditional application to overloaded class
methods.

\subsubsection*{Abstraction}

We have accounted for the uncurried application function.  Now consider its curried equivalent,
expressed in an intuitionistic setting as $\lambda f.\lambda x.f\,x$ of type
$(t \to u) \to t \to u$.  The problem here is similar to the problem with the $K$ combinator or the
curried pair constructor.  Suppose that we apply this function to some value $V$, giving
$\lambda x. V\,x$: whether this function needs to be linear depends on the linearity of $V$.  We
thus have six incomparable types for the curried application function:
\begin{gather*}
  (t \lto u) \lto t \lto u \qquad (t \uto u) \lto t \lto u \qquad (t \uto u) \lto u \uto u \\
  (t \lto u) \uto t \lto u \qquad (t \uto u) \uto t \lto u \qquad (t \uto u) \uto u \uto u
\end{gather*}
Our approach to overloading application allows us to give names to individual function arrows in a
type.  Unfortunately, even this is not sufficient to account for the types of the application
function; it only allows us to reduce the six types above to two:
\[
   (t \labto f u) \to t \labto f u \qquad (t \uto u) \to t \lto u
\]
However, this observation suggests our actual solution.  Consider a more general type, subsuming the
two above (but admitting one erroneous case): $(t \labto f u) \to t \labto g u.$ The first case
above is where $f$ and $g$ are the same type, and the second is where $f$ is less restricted (i.e.,
admits more structural rules) than $g$.  The erroneous case is where $f$ is more restricted (i.e.,
admits fewer structural rules) than $g$.  We introduce a new predicate,
$\tau \moreunlimited \upsilon$, which holds when $\tau$ admits more structural rules than
$\upsilon$.  We can now give the principal type of the application operator:
$f \geq g \then (t \labto f u) \to t \labto g u$.

Our examples have focused on function types.  However, the $\moreunlimited$ relation is not limited
to functions; for example, consider the possible types of the curried pair constructor
$\lambda x. \lambda y. (x, y)$:
\[
   \Unl t \then t \to u \to \pair t u \qquad t \to u \lto \pair t u
\]
As for the application operator, we see that the linearity of the final arrow is restricted by the
types appearing before it in the type signature.  Unlike in that case, however, the earlier type in
question is not a function type.  We can give the curried pair constructor the principal type
$t \geq f \then t \to u \labto f \pair t u$.

\subsection{\lang in Action}\label{sec:inaction}

One of the pleasing aspects of this work has been the simplicity of our motivating examples: the $K$
combinator and application functions are very short, but reveal the unique benefits of \lang.  We
conclude this section by turning to several larger examples.  First, we consider a simple embedding
of dyadic session types, a typical application of linear typing.  Doing so demonstrates that we have
not made our system too permissive.  Second, we consider a presentation of Haskell's monad class and
several of its instances.  This shows that \lang supports the full generality of intuitionistic
functional programming abstractions, and demonstrates the additional information captured by a
linear type system.

For these examples, we will assume various language features present in Haskell, such as new type
definitions, multi-parameter type classes with functional dependencies, and \texttt{do} notation for
monads.  We believe these are representative of realistic settings for linear functional
programming.  However, these ideas are not fundamental to our approach, and our formalization in the
following sections will consider a core calculus that does not assume such language features or
syntactic sugar.

\subsubsection*{Dyadic Session Types}
\newcommand{\st}{\zeta}

Session types, introduced by Honda~\cite{Honda93}, provide a typing discipline for communication
protocols among asynchronous processes.  There is a significant body of work exploring the
combination of session-typed and functional programming.  Much of this work has focused on defining
new linear calculi, combining functional and concurrent
programming~\cite{VasconcelosGR06,GayV10,LindleyM15}.  These calculi frequently include details
specific to session typing in their type systems, and so seem a poor fit for general purpose
programming languages.  Pucella and Tov~\cite{PucellaT08} give an encoding of session types in
Haskell, wrapping an underlying untyped use of channels.  They express the session typing discipline
using the existing features of the Haskell class system.  However, they threat channels as
second-class values, capturing the session types of channels in a parameterized monad~\cite{Atkey09}
rather than in the types of the channels themselves.  One consequence of this is that sending and
receiving channels, while possible, requires primitive operations (with particularly involved types)
distinct from those for sending and receiving values.  We will show that \lang allows us to have the
best of both worlds: because \lang is linear, we can have first-class channels, and because \lang
fits into the existing work on qualified types we can encode the session typing discipline without
having to extend our core type system.

Honda gives five constructors for session types $\st$, interpreted as follows:
\[\begin{array}{ll}
    \gvout \tau \st & \text{Send a value of type $\tau$, then continue as $\st$}  \\
    \gvin \tau \st & \text{Receive a value of type $\tau$, then continue as $\st$} \\
    \st \uplus \st' & \text{Choose between behaviors $\st$ and $\st'$} \\
    \st \nplus \st' & \text{Offer a choice of behaviors $\st$ and $\st'$} \\
    \gvend & \text{No communication}
  \end{array}\]
(Our syntax for the choice constructors differs from Honda's to avoid conflict with the notation for
the linear logic connectives.)  Lindley and Morris~\cite{LindleyM15} observed that, in a linear
functional setting, the choice types can be encoded in terms of $\oplus$ and the input and
output types, and so we omit them from our example.  We introduce types for the remaining session
types---these types are empty, as we will use them as tags rather than to type channels directly.
\begin{code}
data t :!: s
data t :?: s
data End
\end{code}
Honda observed that communicating processes had dual expectations for their shared channels: if one
process expects to send a value of type $\tau$, the other process should expect to receive a value of
type $\tau$.  Following Pucella and Tov~\cite{PucellaT08}, we can capture this using a type class with
functional dependencies~\cite{Jones00}:
\begin{code}
class Dual t u | t -> u, u -> t
instance Dual s s' => Dual (t :!: s) (t :?: s')
instance Dual s s' => Dual (t :?: s) (t :!: s')
instance Dual End End
\end{code}
We now turn to channels and their primitive operators.
\begin{code}
data Ch s
instance Un (Ch End)
\end{code}
Unlike other approaches to encoding session types in functional languages, we treat $\gvend$
channels as unrestricted, avoiding the need for explicit \texttt{close} operations.  Previous work
on linearity has discussed the encapsulation of unrestricted types in linear ones, either via
existential types~\cite{MazurakZZ10,TovP11} or via a module system~\cite{PucellaT08}.
Alternatively, one might prefer to take the notion of linear channels as primitive.  Either approach
is possible in \lang; as we are primarily concerned with the use of linear types, we omit further
discussion of them here.  (But see the extended version of this paper~\cite{Morris16}
for the details of the packaging approach.)  The primitive operations on session-typed channels are
as follows:
\begin{code}
fork    :: Dual s s' => (Ch s -> M ()) -> M (Ch s')
send    :: t >= f => t -> Ch (t :!: s) -f> M (Ch s)
receive :: Ch (t :?: s) -> M (t :*: Ch s)
\end{code}
We adopt the \texttt{fork} construct of Lindley and Morris both for its simplicity and because it
assures deadlock freedom.  The \texttt{Dual} predicate assures that the session types \texttt{s} and
\texttt{s\textquotesingle} are well-formed and dual.  Gay and Vasconcelos~\cite{GayV10} give two
typings for the \texttt{send} function, depending on the linearity of its first argument:
\[\begin{array}{ll}
    t \uto (t \mathbin ! s) \uto s &\text{if $t$ is unrestricted} \\
    t \uto (t \mathbin ! s) \lto s &\text{otherwise}
  \end{array}\]
This fits precisely the pattern captured by the $\moreunlimited$ predicate in \lang.  Finally, as
the communication primitives are side-effecting, we assume the results are embedded in some suitable
monad \texttt{M}.  (This is not an entirely innocuous choice; we will return to monads in a linear
setting for our next example.)

We present a simple example using session-typed channels.  We begin with a process that performs an
arithmetic operation:
\begin{code}
multiplier c =
  do (x, c) <- receive c
     (y, c) <- receive c
     send (x * y) c
     return ()
\end{code}
The \texttt{multiplier} function defines a process that expects to read two numbers on channel
\texttt{c}, and then sends their product back along the same channel.  The inferred type for
\texttt{multiplier} is \lstinline~Num t => Ch (t :?: (t :?: (t :!: End))) -> M ()~.  Note that,
despite our reuse of the name \texttt{c}, each call to a communication primitive returns a new copy
of the channel, which is used linearly.  Next, we define a process to communicate with
\texttt{multiplier}.  To illustrate the use of channels as first-class values, we define it in a
round-about way.  First, we define a process that provides only one of the two expected values:
\begin{code}
sixSender c =
  do (d, c) <- receive c
     send 6 d
     send d c
     return ()
\end{code}
This function defines a process that begins by receiving a channel \texttt{d} along \texttt{c}; it
then sends 6 along the received channel before returning the received channel along \texttt{c}.
Thus, its type is
\begin{code}
Num t => Ch (Ch (t :!: s) :?: Ch s :!: End) -> M ()
\end{code}
Finally, we can define the main process, which uses the preceding processes to compute 42:
\begin{code}
answer = do d <- fork sixSender
            c <- fork multiplier
            d <- send c d
            (c, d) <- receive d
            c <- send 7 c
            (x, c) <- receive c
            return x
\end{code}
This example demonstrates the advantages of \lang for linear programming.  Unlike encoding-based
approaches, we have simple types and uniform treatment of channels and other data.  Unlike other
concurrency-focused approaches, we have not built any aspects of session typing into our language or
its type system.

\subsubsection*{Monads}
\newcommand{\bind}{\ensuremath{\mathbin{>\!\!>\!\!=}}}

In the previous example, we assumed that we could express our communication primitives monadically,
to account for their side effects.  As they are fundamentally reliant on higher-order functions, it
is worth examining the interaction between linearity and the monadic combinators.  For a simple
example, consider the desugaring of \texttt{answer}, which begins
\[
  \mathtt{fork\;sixSender} \bind \bs d \to \mathtt{fork\;multiplier} \bind \bs c \to M
\]
where $M$ denotes the remainder of \texttt{answer}, and both $c$ and $d$ are free in $M$.  As $d$ is
of linear type, we see that $\lambda c. M$ must be a linear function. Does this mean that the result
of $\bind$ must also be linear?  How does this play out for other monads, like the \texttt{Maybe}
monad?

Of course, we could transport standard intuitionistic definitions of monads directly into \lang,
treating all functions as unrestricted.  Doing so would allow us to use monads for unrestricted
values without any new complexity.  However, doing so would also rule out interesting cases, such
as those with channels in the previous example.  Here we take the opposite perspective, attempting
to generalize standard notions of monads to include the linear cases.  We will consider two
canonical examples, failure and state.

First, we consider failure.  We assume we have some type $\mathtt{Maybe}\,t$ with constructors
\texttt{Just} and \texttt{Nothing}; observe that $\mathtt{Maybe}\,t$ is unrestricted precisely when
$t$ is unrestricted.  To demonstrate that \texttt{Maybe} is a monad, we give implementations of the
\texttt{return} and $(\bind)$ operators, as follows:
\begin{code}
return = \x -> Just x
(>>=)  = \m -> \f -> case m of
                       Nothing -> Nothing
                       Just x -> f x
\end{code}
The typing of \texttt{return} is uninteresting.  On the other hand, consider the use of $f$ in the
body of $(\bind)$: if $m$ is \texttt{Nothing}, then $f$ is discarded, whereas if $m$ is
$\mathtt{Just}\,x$, then $f$ is used once.  So, we see that $f$ must be unrestricted, and so we have
the types:
\[\begin{array}{l@{\;::\;}l}
  \mathtt{return} & t \to \texttt{Maybe} \, t \\
  (\bind) & t \geq f \then \texttt{Maybe} \, t \to (t \uto \texttt{Maybe}\,u) \labto f \texttt{Maybe}\,u \\
  \end{array}\]
The requirement that $f$ be unrestricted captures that the remainder of the computation may not
occur, an important characteristic of the failure monad.  For example, this means that the monad !M!
in the session types example cannot include exceptions.  This should align with our expectations: if
a process fails, it cannot fulfill its outstanding session-typed obligations.

Next, we consider the state monad.  A state monad for state values of type \texttt{S} is typically
implemented in Haskell by the type $\mathtt{S} \to (t, \mathtt{S})$.  This introduces additional
choice in the linear case: should we consider values of type $\mathtt{S} \lto \pair t {\mathtt{S}}$
or of type $\mathtt{S} \uto \pair t {\mathtt{S}}$?  What constraints would this choice impose on the
use of the monad?  We can clarify these questions by considering the definition of \texttt{return}
and $(\bind)$.  (Relying on our generalization of abstraction and application, we consider these
implementations in parallel with the choice of the state monad itself.)
%
%
\begin{code}
return = \x -> \s -> (x, s)
(>>=)  = \m -> \f -> \s -> let (x, s') = m s in f x s
\end{code}
We make two observations about the state monad.  First, $x$ is captured in $\bs s \to (x, s)$;
therefore, a state computation can only be as unrestricted as its result values.  (This is true of
the failure monad as well, but is reflected in the inherent linearity of \texttt{Maybe} types.)
Second, note that the function $f$ is used linearly in the body of $(\bind)$, so its type need not
be unrestricted (unlike for the failure monad).  These observations are reflected in the types of
\texttt{return} and $(\bind)$.  We begin by introducing an alias for the state monad type:
\begin{code}
type State k s t = s -k> (t :*: s)
\end{code}
We can then type \texttt{return} and $(\bind)$ by
\begin{code}
return :: t >= State k s t => t -> State k s t
(>>=)  :: (State k s t >= g, f >= State k s u) =>
          State k s t -> (t -f> State k s u) -g>
          State k s u
\end{code}
The predicate $\texttt{State}\,k\,s\,t \geq g$ reflects that the term $\bs f \to \dots$
has captured $m$ of type $\texttt{State}\,k\,s\,t$.


Finally, we generalize these examples.  The problem is the type of the second argument to $(\bind)$:
to be useful in the linear context, we must sometimes include the restricted function type, but to
incorporate the full range of monads we must sometimes limit it to unrestricted functions.  We
encompass both cases using a multi-parameter type class for monads:
\begin{code}
class Monad f m | m -> f where
  return :: t >= m t => t -> m t
  (>>=)  :: (m t >= g, f >= m u) =>
            m t -> (t -f> m u) -g> m u
\end{code}
The definitions above give instances of our new \texttt{Monad} class:
\begin{code}
instance Monad (->*) Maybe
instance Monad k (State k s)
\end{code}
and that the example of dyadic session types will type in monads $m$ such that
$\texttt{Monad}\;(\lto)\;m$ is provable.

We should emphasize that, because $f$ is functionally dependent on $m$, our reformulation of the
\lstinline!Monad! class does not introduce any new polymorphism, or new potential for ambiguity.
Rather, it makes explicit (at the type level) existing differences in the composition of monadic
computations.


\section{Substructural Qualified Types}\label{sec:typesystem}

We have considered some of the challenges of using linear calculi in practice, given an intuitive
description of how we addresses these challenges using qualified types, and demonstrated how our
solution might be realized in a Haskell-like practical programming language.  In this section, we
give a formal account of our approach to substructural qualified types.  We begin by giving an
overview of a core \lang calculus and its type system~\secref{typing}.  We then give a
syntax-directed variant on the type system~\secref{styping}, preparatory to giving an
Algorithm~$\mathcal{M}$~\cite{LeeY98} style type inference algorithm~\secref{inference}.  Finally,
we relate \lang typing to typing for a non-substructural core calculus~\secref{conservativity},
making concrete our claims that \lang encompasses existing functional programming practice.

\subsection{\lang Terms and Typing}\label{sec:typing}
\newcommand{\T}{\mathcal{T}}

\begin{figure}
\[\begin{array}{ll@{\hspace{3mm}}ll}
  \text{Term variable} & x,y \in Var & \text{Type variables} & t,u \in TVar \\
  \text{Multienvironments} & \Eta & \text{Environments} & \Gamma,\Delta \\
  \text{Type constructors} & \multicolumn{3}{l}{T^\kappa \in \T^\kappa \text{ where $\Set{\oplus,\uto,\lto} \subseteq \T^{\star \to \star \to \star}$}}
\end{array}\]
\begin{longsyntax}
  \text{Kinds} & \tcr{\kappa} & ::= & \star \mid \kappa \to \kappa \\
  \text{Types} & \tcr{\tau^\kappa} & ::= & t \mid T^\kappa \mid \tau^{\kappa' \to \kappa}\,\tau^{\kappa'} \\
  \tcl{\text{Predicates}} & \pi & ::= & \Unl\tau \mid \Fun\tau \mid \tau \moreunlimited \upsilon \\
  \tcl{\text{Qualified types}} & \rho & ::= & \tau^\star \mid \pi \then \rho \\
  \tcl{\text{Type schemes}} & \sigma & ::= & \rho \mid \forall t. \sigma \\
  \text{Expressions} & \tcr{M,N} & ::= & x \mid K\,M \mid \lambda x. M \mid M\,N \mid \inl{M} \mid \inr{N} \\ 
   &&& \mid & \case{M}{x}{N}{y}{N'} \\
   &&& \mid & \mlet{x}{M}{N} \mid \mlet{K\,x}{M}{N}
\end{longsyntax}
\caption{\lang types and terms.}\label{fig:syntax}
\end{figure}

\newcommand{\Without}[2]{#1}

\begin{figure*}
\begin{gather*}
\fbox{$P \mid \Eta \vdash M: \sigma$}
\\
\infbox{\irule[\trule{var}];{P \mid \Bag{x:\sigma} \vdash x:\sigma}}
\isp
\infbox{\irule[\trule{\C}]
              {P \mid \Eta,\Eta',\Eta' \vdash M: \sigma}
              {P \vdash \Junl{\Eta'}};
              {P \mid \Eta,\Eta' \vdash M: \sigma}}
\isp
\infbox{\irule[\trule{\W}]
              {P \mid \Eta \vdash M: \sigma}
              {P \vdash \Junl{\Eta'}};
              {P \mid \Eta,\Eta' \vdash M: \sigma}}
\\
\infbox{\irule[\trule{\I\to}]
              {\begin{array}{c}
                {P \mid \Without \Eta x, x:\tau \vdash M: \upsilon} \\[2px]
                {P \entails \Fun \phi} \isp
                {P \vdash \Eta \moreunlimited \phi}
               \end{array}};
              {P \mid \Without \Eta x \vdash \lambda x.M: \phi \, \tau \, \upsilon}}
\isp
\infbox{\irule[\trule{\E\to}]
              {\begin{array}{c}
                  {P \mid \Eta \vdash M: \phi \, \tau \, \upsilon} \\[2px]
                  {P \mid \Eta' \vdash N: \tau} \isp
                  {P \entails \Fun \phi}
               \end{array}};
              {P \mid \Eta,\Eta' \vdash M\,N: \upsilon}}
\isp
\infbox{\irule[\trule{let}]
              {P \mid \Eta  \vdash M: \sigma}
              {P \mid \Eta', x: \sigma \vdash N:\tau};
              {P \mid \Eta,\Eta' \vdash \mlet{x}{M}{N}: \tau}}
\\
\infbox{\irule[\trule{\I\oplus$_i$}]
              {P \mid \Eta \vdash M: \tau_i};
              {P \mid \Eta \vdash \ini{M}\vdash \opt{\tau_1}{\tau_2}}}
\isp
\infbox{\irule[\trule{\E\oplus}]
              {\begin{array}{c}
                 P \mid \Eta \vdash M: \opt{\tau_1}{\tau_2} \isp
                 P \mid \Eta'_{x},x:\tau_1 \vdash N: \upsilon \isp
                 P \mid \Eta'_{x},x:\tau_2, \vdash N': \upsilon
               \end{array}};
              {P \mid \Eta,\Eta'_{x} \vdash \case{M}{x}{N}{x}{N'}: \upsilon}}
\\
\infbox{\irule[\trule{make}]
              {\begin{array}{c}
                 {K: (\forall \seq t. \exists \seq u. Q \then \tau') \uto \tau} \isp
                 {P \entails [\seq{\upsilon}/\seq{u}]Q} \\[2px]
                 {P \mid \Eta \vdash M: [\seq{\upsilon}/\seq{u}]\tau'} \isp
                 {\seq u \notin ftv(P,\Eta)}
               \end{array}};
              {P \mid \Eta \vdash K\,M : \tau}}
\isp
\infbox{\irule[\trule{break}]
              {\begin{array}{c}
                 {K: (\forall \seq t. \exists \seq u. Q \then \tau') \uto \tau} \isp
                 {P \mid \Eta \vdash M: \tau} \\[2px]
                 {P,[\seq{\upsilon}/\seq t]Q \mid \Eta',x:[\seq\upsilon/\seq t]\tau' \vdash N: \upsilon'} \isp
                 {\seq u \not\in ftv(P,\Eta,\Eta',\upsilon')}
               \end{array}};
              {P \mid \Eta,\Eta' \vdash \mlet{K\,x}{M}{N} : \upsilon'}}
\\
\infbox{\irule[\trule{\I\then}]
              {P,\pi \mid \Eta \vdash M: \rho};
              {P \mid \Eta \vdash M: \pi \then \rho}}
\isp
\infbox{\irule[\trule{\E\then}]{P \mid \Eta \vdash M: \pi \then \rho}
                               {P \entails \pi};
                               {P \mid \Eta \vdash M: \rho}}
\isp
\infbox{\irule[\trule{\I\forall}]
              {P \mid \Eta \vdash M: \sigma}
              {t \notin ftv(P,\Eta)};
              {P \mid \Eta \vdash M: \forall t. \sigma}}
\isp
\infbox{\irule[\trule{\E\forall}]
              {P \mid \Eta \vdash M: \forall t. \sigma};
              {P \mid \Eta \vdash M: [\tau/t]\sigma}}
\\[5px]
\begin{gathered}
\fbox{$P \vdash \Junl\cdot$}
\\
\infbox{\irule[\trule{un-$\tau$}]
              {P \entails \Unl\tau};
              {P \vdash \Junl\tau}}
\isp
\infbox{\irule[\trule{un-$\rho$}]
              {P, \pi \vdash \Junl\rho};
              {P \vdash \Junl{\pi \then \rho}}}
\\
\infbox{\irule[\trule{un-$\sigma$}]
              {P, \Unl t \vdash \Junl\sigma};
              {P \vdash \Junl{\forall t.\sigma}}}
\isp
\infbox{\irule[\trule{un-$\Eta$}]
              {\bigwedge_{x:\sigma \in \Eta} P \vdash \Junl\sigma};
              {P \vdash \Junl\Eta}}
\end{gathered}
\isp\isp
\begin{gathered}
\fbox{$P \vdash \cdot \moreunlimited \phi$}
\\
\infbox{\irule[\trule{$\moreunlimited$-$\tau$}]
              {P \entails \tau \moreunlimited \phi};
              {P \vdash \tau \moreunlimited \phi}}
\isp
\infbox{\irule[\trule{$\moreunlimited$-$\rho$}]
              {P, \pi \vdash \rho \moreunlimited \phi};
              {P \vdash (\pi \then \rho) \moreunlimited \phi}}
\\
\infbox{\irule[\trule{$\moreunlimited$-$\sigma$}]
              {P, \Unl t \vdash \sigma \moreunlimited \phi};
              {P \vdash (\forall t. \sigma) \moreunlimited \phi}}
\isp
\infbox{\irule[\trule{$\moreunlimited$-$\Eta$}]
              {\bigwedge_{x:\sigma \in \Eta} P \vdash \sigma \moreunlimited \phi};
              {P \vdash \Eta \moreunlimited \phi}}
\end{gathered}
\end{gather*}
\caption{Typing rules.}\label{fig:typing}
\end{figure*}


The syntax of \lang types and terms is shown in \figref{syntax}.  \lang types are stratified
according to a simple kind system; we write $\tau$, $\upsilon$ and $\phi$ (without superscripts) to
range over types of any kind.  (Unlike Mazurak et al.~\cite{MazurakZZ10}, we use kind $\star$ for
all types, not just unrestricted ones.) We assume that $\lto$, $\uto$ and $\oplus$ are binary type
constructors, which we will write infix, corresponding to linear and unrestricted functions and
additive sums.  We do not include multiplicative or additive products ($\pair \tau \upsilon$ and
$\tau \with \upsilon$), as these can be encoded in terms of the other types.  (These encodings
depend on our overloading of abstraction for their full generality.)  We allow arbitrary additional
type constructors, providing other (user-defined) data types.  Data types capture first-class
universal and existential types, following the approach of Jones's FCP~\cite{Jones97}.  While we
have not used these features in our examples, we include them for two reasons.  First, existential
types are used prominently in other approaches to linear functional
programming~\cite{MazurakZZ10,TovP11}, particularly to construct linear wrappers around unrestricted
types, and so we show that Quill accommodates similar constructions.  Second, existentials provide a
another application of the techniques developed to account for the linearity of functions; we
describe this at more length when consider extensions to our core calculus~\secref{extensions}.
Predicates $\pi$ include those necessary for our treatment of linearity (and can constrain
higher-kinded types).  Qualified types and type schemes are standard for overloaded Hindley-Milner
calculi.  We write $\forall \seq t.Q \then \tau$ to abbreviate
$\forall t_1 \dots \forall t_n.P_1 \then \dots \then P_m \then \tau.$

\lang includes standard terms for variables, abstractions, applications, and (additive) sums.  We
introduce polymorphism at \mkwd{let} bindings.  First-class existential and universal types are
expressed using constructors $K$.  We assume an ambient signature mapping individual constructors
$K$ to types $\forall \vec{v}. (\forall \seq t. \exists \seq u. Q \then\tau') \uto \tau$.  (This
type is not included in $\sigma$; $\sigma$ denotes inferable type schemes.)  Construction $K\,M$
builds a value of type $\tau$, assuming that $\tau'$ has a suitably generic type.  We insist that
constructors be fully applied.  Deconstruction $\mlet{K\,x}{M}{N}$ eliminates such values.  The
introduction of first-class polymorphism through data types corresponds to common practice in
Haskell, and allows us to make clear the extent of type inference.

\figref{typing} gives the \lang type system.  The typing judgment is $P \mid \Eta \vdash M: \sigma$,
where $P$ is a set of predicates on the type variables in the remainder of the judgment, $\Eta$ is a
typing environment, $M$ is a term and $\sigma$ a type scheme.  Our treatment of the typing
environment follows a standard approach for linear logic, but differs from some of the existing work
on linear type systems.  We use multisets of typing assumptions (which we call
``multienvironments''); thus, the multienvironment $\Bag{x:\sigma,x:\sigma}$ is distinct from
$\Bag{x:\sigma}$.  We require that multienvironments be consistent, so if
$\Bag{x:\sigma,x:\sigma'} \subseteq H$ then we must have $\sigma = \sigma'$, and we write $H,H'$ for
the consistent multiset union of $H$ and $H'$.  Assumptions of unrestricted type are duplicated in
\trule\C and discarded in \trule\W; both rules use the auxiliary judgment $P \vdash \Junl \cdot$,
which lifts the \texttt{Un} predicate to typing environments.  (The assumption $\Unl t$ in
\trule{un-$\sigma$} accounts for terms like the empty list, which should be treated as unrestricted
until their type variables are instantiated.)  This allows us to simplify the other rules: in
\trule{var}, there must only be one binding in the environment, while rules like \trule{\E\to} can
split the typing environment among their hypotheses.  In particular, we avoid introducing an
auxiliary judgment to split type environments while sharing assumptions of unrestricted type present
in many linear type systems~\cite{MazurakZZ10,TovP11}.  Rules \trule{\I\to} and \trule{\E\to}
implement overloading of abstraction and application.  In \trule{\E\to}, note that we allow the
function term to be of any type $\phi\,\tau\,\upsilon$, so long as it satisfies the constraint
$\Fun\phi$.  In \trule{\I\to}, we allow a term $\lambda x.M$ to have any function type, so long as
that type is more restricted than its environment; the auxiliary judgment
$P \vdash \cdot \geq \cdot$ lifts the $\geq$ predicate to type environments.  We will assume
throughout this presentation that binders introduce fresh names.  First-class polymorphism is
introduced in \trule{make} and eliminated in \trule{break}; our approach follows
Jones~\cite{Jones97} almost exactly, but adds a predicate context $Q$.  We write
$K : (\forall t. \exists u. Q \then \tau') \to \tau$ to denote an instantiation of the signature for
$K$ in the ambient context.  We assume that each data type has at most one constructor; more complex
data types can be expressed using the other features of the type system.  The remaining rules are
standard for linear sums and qualified polymorphism.

\begin{figure}
\begin{gather*}
\infbox{\irule{P \ni \pi};{P \entails \pi}}
\isp
\infbox{\irule{\bigwedge_{\pi \in Q} P \entails \pi};
              {P \entails Q}}
\isp
\infbox{\irule{\tau = {\lto} \lor \tau = {\uto}};
              {P \entails \Fun\tau}}
\\
\infbox{\irule{K : (\forall \seq t. \exists \seq u. Q \then \tau') \uto \tau}
              {P, Q, \Unl{\seq t} \entails \Unl{\tau'}};
              {P \entails \Unl\tau}}
\\
\infbox{\irule;{P \entails \Unl{(\tau \uto \upsilon)}}}
\isp
\infbox{\irule{P \entails \Unl{\tau_1}}
              {P \entails \Unl{\tau_2}};
              {P \entails \Unl{(\opt{\tau_1}{\tau_2})}}}
\\
\infbox{\irule{P \entails \Unl\tau};{P \entails \tau \moreunlimited (\upsilon \uto \upsilon')}}
\isp
\infbox{\irule;{P \entails \tau \moreunlimited (\upsilon \lto \upsilon')}}
\\
\infbox{\irule{P \entails \tau \moreunlimited \phi\,t}
              {t \text{ fresh}};
              {P \entails \tau \moreunlimited \phi}}
\isp
\infbox{\irule{P \entails \tau\,t \moreunlimited \phi}
              {t \text{ fresh}};
              {P \entails \tau \moreunlimited \phi}}
\end{gather*}
\caption{Entailment rules.}\label{fig:entail}
\end{figure}

\figref{entail} gives a minimal definition of the predicate entailment relation $P \entails Q$.  One
strength of type systems based on qualified types is that the predicate system provides a natural
point of extension, and our approach here is no different.  Nevertheless, we specify some rules for
the entailment judgment, namely the linearity of the built-in types and that (only) $\lto$ and
$\uto$ are in class $\mathtt{Fun}$.  In determining the linearity of a data type $\tau$, we assume
that the universally quantified type variables $\seq t$ are unrestricted (as a term of type $\tau$
cannot have made any assumptions of $\seq t$), but cannot do so for the existentially quantified
variables $\seq u$, as they may have been instantiated arbitrarily in constructing the $\tau$ value.
We have intentionally given a minimal specification of $\geq$; in particular, we have omitted
various simplification rules which might be expected in a practical implementation, and have limited
the our attention to cases of $\geq$ with function types on their right-hand side (as those are the
only such predicates introduced by our typing rules).  We will return to the definition of this
class when we discuss extensions to \lang~\secref{extensions}.  Finally, the lifting cases for
$\geq$ are a notational convenience; for example, they allow us to write $\tau \geq \phi$ rather
than $\tau \geq \phi\,t\,u$ for fresh $t$ and $u$.

\subsection{A Syntax-Directed \lang Type System}\label{sec:styping}

The \lang type system has a number of rules that are not syntax directed, including the structural
rules and the rules introducing and eliminating polymorphism.  To simplify the definition of type
inference and the proofs of its correctness, we give a syntax-directed variant of the \lang type
system.  In doing so, we address two independent concerns.  First, the rules \trule{\I\forall},
\trule{\E\forall}, \trule{\I\then}, and \trule{\E\then} may be used at any point in a derivation.
This problem has already been studied in the general context of qualified types.  An identical
solution applies in \lang: uses of \trule{\E\forall} and \trule{\E\then} may always be permuted to
occur at occurrences of \trule{var}, while uses of \trule{\I\forall} and \trule{\I\then} may always
be permuted to occur at occurrences of \trule{let} or at the end of the derivation.  Second, the
structural rules \trule\W and \trule\C may also appear at any point in a typing derivation.  As in
the polymorphism cases, we show that uses of these rules can be permuted to definite places in the
derivation: uses of \trule\C can be permuted to appear immediately below a rule with multiple
hypotheses (such as \trule{\E\to} or \trule{\E\oplus}) and uses of \trule\W can be permuted to
occurrences of \trule{var}.

\begin{figure*}
\begin{gather*}
\infbox{\irule[\strule{var}]
              {P \vdash \Junl\Delta}
              {(P \then \tau) \lessgeneral \sigma};
              {P \mid \Delta,x:\sigma \vdashS x: \tau}}
\isp
\infbox{\irule[\strule{let}]
              {\begin{array}{c}
                  Q \mid \Gamma,\Delta \vdashS M: \tau \isp
                  {\sigma = Gen(\Gamma,\Delta;Q \then \tau)} \\[2px]
                  P \mid \Without {\Gamma'} x,\Without{\Delta}{x},x:\sigma \vdashS N:\upsilon \isp
                  P \vdash \Junl\Delta
               \end{array}};
              {P \mid \Gamma,\Without {\Gamma'} x,\Delta \vdashS \mlet{x}{M}{N}: \upsilon}}
\\
\infbox{\irule[\strule{\I\to}]
              {\begin{array}{c}
                 P \mid \Without \Gamma x, x:\tau \vdashS M: \upsilon \\
                 P \entails \Fun \phi \isp
                 P \vdash \Gamma \moreunlimited \phi \isp
               \end{array}};
              {P \mid \Without \Gamma x \vdashS \lambda x.M: \phi \, \tau \, \upsilon}}
\isp
\infbox{\irule[\strule{\E\to}]
              {\begin{array}{c}
                 P \mid \Gamma,\Delta \vdashS M: \phi \, \tau \, \upsilon \isp
                 P \mid \Gamma',\Delta \vdashS N: \tau \\[2px]
                 P \entails \Fun \phi \isp
                 P \vdash \Junl\Delta
               \end{array}};
              {P \mid \Gamma,\Gamma',\Delta \vdashS M\,N: \upsilon}}
\\
\infbox{\irule[\strule{\I\oplus$_i$}]
              {P \mid \Gamma \vdashS M: \tau_i};
              {P \mid \Gamma \vdashS \ini{M}: \opt{\tau_1}{\tau_2}}}
\isp
\infbox{\irule[\strule{\E\oplus}]
              {\begin{array}{c}
                 P \mid \Gamma,\Delta \vdashS M: \opt{\tau_1}{\tau_2} \isp
                 P \vdash \Junl\Delta \\
                 P \mid \Without {\Gamma'} x,\Without \Delta x,x:\tau_1 \vdashS N: \upsilon \isp
                 P \mid \Without {\Gamma'} x,\Without \Delta x,x:\tau_2 \vdashS N': \upsilon
               \end{array}};
              {P \mid \Gamma,\Without {\Gamma'} x,\Delta \vdashS \case{M}{x}{N}{x}{N'}: \upsilon}}
\\
\infbox{\irule[\trule{make}]
              {\begin{array}{c}
                 {K: (\forall \seq t. \exists \seq u. Q \then \tau') \uto \tau} \isp
                 {P \entails [\seq\upsilon/\seq u]Q} \\[2px]
                 {P \mid \Gamma \vdashS M: [\seq\upsilon/\seq u]\tau'} \isp
                 {\seq t \notin ftv(P,\Gamma)}
               \end{array}};
              {P \mid \Gamma \vdashS K\,M : \tau}}
\isp
\infbox{\irule[\trule{break}]
              {\begin{array}{c}
                 {K: (\forall \seq t. \exists \seq u. Q \then \tau') \uto \tau} \isp
                 {P \mid \Gamma,\Delta \vdashS M: \tau} \isp
                 {P \vdashS \Junl\Delta} \\[2px]
                 {P,[\seq\upsilon/\seq t]Q \mid \Gamma', \Without \Delta x,x:[\seq\upsilon/\seq t]\tau' \vdashS N: \upsilon'} \isp
                 {\seq u \not\in ftv(P,\Gamma,\Gamma',\Delta,\upsilon')}
               \end{array}};
              {P \mid \Gamma,\Without {\Gamma'} x,\Delta \vdashS \mlet{K\,x}{M}{N} : \upsilon'}}
\end{gather*}
\caption{Syntax-directed typing rules.}\label{fig:styping}
\end{figure*}


\figref{styping} gives the syntax-directed variant of the \lang system.  The judgment
$P \mid \Gamma \vdashS M: \tau$, is a syntax-directed variant of $P \mid \Eta \vdash M : \sigma$,
and uses standard type environments $\Gamma$ rather than multienvironments $\Eta$.  We write
$\Gamma,\Gamma'$ to denote partitioning the environment; the treatment of contraction and weakening
is explicit in the typing rules, rather than via a partitioning relation on typing environments.
The auxiliary judgments are unchanged.  The syntax-directed system differs from the original type
system in two ways.  First, we account for polymorphism.  Our approach is identical to Jones's
approach for (intuitionistic) qualified types~\cite{Jones94}: we introduce instantiation and
generalization operators, accounting for the role of predicates, and collapse the treatment of
polymorphism into the instances of \strule{var} and \strule{let}.
\begin{defn}
  We define instantiation and generalization as follows:
  \begin{enumerate}
  \item Let $\sigma$ be some type scheme $\forall \seq t. P \then \tau'$.  We say that
    $Q \then \tau$ is an instance of $\sigma$, written $(Q \then \tau) \lessgeneral \sigma$, if
    there is some $\seq\upsilon$ such that $\tau = [\vec\upsilon/\vec t]\tau'$ and
    $Q \entails [\vec\upsilon/\vec t]P$.
  \item Let $\Gamma$ be a typing environment, and $\rho$ a qualified type.  We define
    $Gen(\Gamma,\rho)$ to be the type scheme $\forall (ftv(\rho) \setminus ftv(\Gamma)). \rho$.
  \end{enumerate}
\end{defn}
\noindent%
We use instantiation in \strule{var}, collapsing a use of \trule{var} and subsequent uses of
\trule{\E\forall} and \trule{\E\then}, and generalization in \strule{let} collapsing a use of
\trule{let} and preceding uses of \trule{\I\then} and \trule{\I\forall}.  Second, we account for
contraction and weakening.  In \strule{var}, we allow an arbitrary environment, so long as the
unused assumptions $\Delta$ are unrestricted.  In \strule{\E\to}, we partition the input environment
into three parts: $\Gamma$ is used exclusively in typing $M$, $\Gamma'$ is used exclusively in
typing $N$, and $\Delta$ is used in both; consequently, assumptions in $\Delta$ must be
unrestricted.  The remaining rules follow the same pattern.

The goal of the syntax-directed type system is a one-to-one correspondence between syntactic forms
and typing rules.  However, it is not the case that a typeable term has exactly one syntax-directed
typing derivation.  For example, while the contents of the linear environments are determined by the
term structure, the contents of the unrestricted environments are not.  For another example,
consider the term $(\lambda x. x) \, y$.  We can choose to type the abstraction as either $t \lto t$
or $t \uto t$ (or even as $t \labto f t$ assuming the predicate $\Fun f$).  Each of these choices
would make the term well-typed, and we assume that the terms obey the same laws (i.e., the choice of
type introduces no observable distinction in the semantics of the term).

We now relate our original and syntax-directed type systems.  We start with environments.
Intuitively, the syntax-driven system introduces contraction when needed, guarded by $\Unl\!\!$
constraints; however, a multienvironment $\Eta$ could contain multiple instances of assumptions with
linear types.  We introduce a notion of approximation between multienvironments $H$ and environments
$\Gamma$ that holds when the only repeated assumptions in $H$ are for unrestricted types.
\begin{defn}
  If $\Eta$ is a multienvironment, $\Gamma$ is an environment, and $P$ is some context, then we say
  that $\Gamma$ approximates $\Eta$ under $P$, written $P \vdash \Eta \approx \Gamma$, if
  $x:\sigma \in \Eta$ if and only if $x:\sigma \in \Gamma$, and if
  $\Bag{x:\sigma,x:\sigma} \subseteq \Eta$, then $P \vdash \Junl\sigma$.
\end{defn}

We now turn to our primary results.  First, derivations in the syntax-directed system correspond to
derivations in the original system.
\begin{theorem}[Soundness of $\vdashS$]\label{thm:soundness-s}
  If $P \mid \Gamma \vdashS M: \tau$ and $P \vdash \Eta \approx \Gamma$, then $P \mid \Eta \vdash M:
  \tau$.
\end{theorem}
\noindent%
The proof is by structural induction on the derivation of $P \mid \Gamma \vdashS M: \tau$, and
relies on introducing instances of the structural and polymorphism rules.

Second, we show completeness of the syntax directed system.  A derivation in the original system may
end with uses of \trule{\I\then} or \trule{\I\forall}, moving predicates from the context to the
type or quantifying over free type variables.  In contrast, there are no such steps in a derivations
in the syntax-directed system. To account for this difference, we introduce a notion of qualified
type schemes, again following Jones~\cite{Jones94}.
\begin{defn}
  A qualified type scheme $(P \mid \sigma)$ pairs a type scheme $\sigma$ with a set of predicates
  $P$.  Let $\sigma$ be $\forall \seq t. Q \then \tau$ and $\sigma'$ be
  $\forall \seq {t'}. Q' \then \tau'$.  We say that $(P \mid \sigma)$ is an instance of
  $(P' \mid \sigma')$, written $(P \mid \sigma) \lessgeneral (P' \mid \sigma')$ iff there are
  $\seq\upsilon$ such that $\tau = [\upsilon_i/t'_i]\tau'$ and
  $P,Q \entails P',[\upsilon_i/t'_i]Q'$.  We treat type schemes $\sigma$ as abbreviations for
  qualified type schemes $(\emptyset \mid \sigma)$.
\end{defn}
\noindent%
We can now state the completeness of the syntax-directed system.
\begin{theorem}[Completeness of $\vdashS$]\label{thm:completeness-s}
  If $P \mid \Eta \vdash M: \sigma$ and $P \vdash \Eta \approx \Gamma$, then there are some $Q$ and
  $\tau$ such that $Q \mid \Gamma \vdashS M: \tau$ and
  $(P \mid \sigma) \lessgeneral Gen(\Gamma,Q \then \tau).$
\end{theorem}
\noindent%
Intuitively, this states that for any derivation in our original type system, there is a derivation
of at least as general a result in the syntax-directed system.  The proof is by induction on the
derivation of $P \mid \Eta \vdash M: \sigma$.  The interesting cases rely on the role of
generalization and instantiation in the syntax-directed type system and the safe movement of
structural rules up derivation trees.

\subsection{Type Inference for \lang}\label{sec:inference}

Having defined a suitable target type system, we can give a type inference algorithm for \lang.  We
have three separate concerns during type inference.  First, we use a standard Hindley-Milner
treatment of polymorphism.  Second, we introduce \texttt{Un} predicates for non-linear use of
variables.  We track the variables used in each expression, and so detect when variables are reused
or discarded.  Third, we account for first-class polymorphism.  We introduce a distinction between
rigid and flexible type variables; only the latter are bound in unification.  These three concerns
add apparent complexity to the type inference algorithm, but can be understood separately.

\newcommand{\Wrule}[1]{\multicolumn{2}{@{}l}{#1}}
\newcommand{\Wgiven}[1]{\ \ & \begin{array}{@{}r@{\hspace{4px}}r@{\hspace{4px}}l@{}} #1 \end{array}}
\newcommand{\Where}{\text{where} &}
\renewcommand{\And}{&}
\newcommand{\Break}{&&\qquad}

\newcommand{\RestrictEnv}[2]{#1|_{#2}}

\newcommand{\M}{\mathcal{M}}
\newcommand{\Mapp}[5]{\M(#1, #2; #3 \vdash #4 : #5)}
\newcommand{\Mres}[3]{#1,#2,#3}
\newcommand{\Unif}[3]{Mgu_{#1}(#2, #3)}

\begin{figure*}
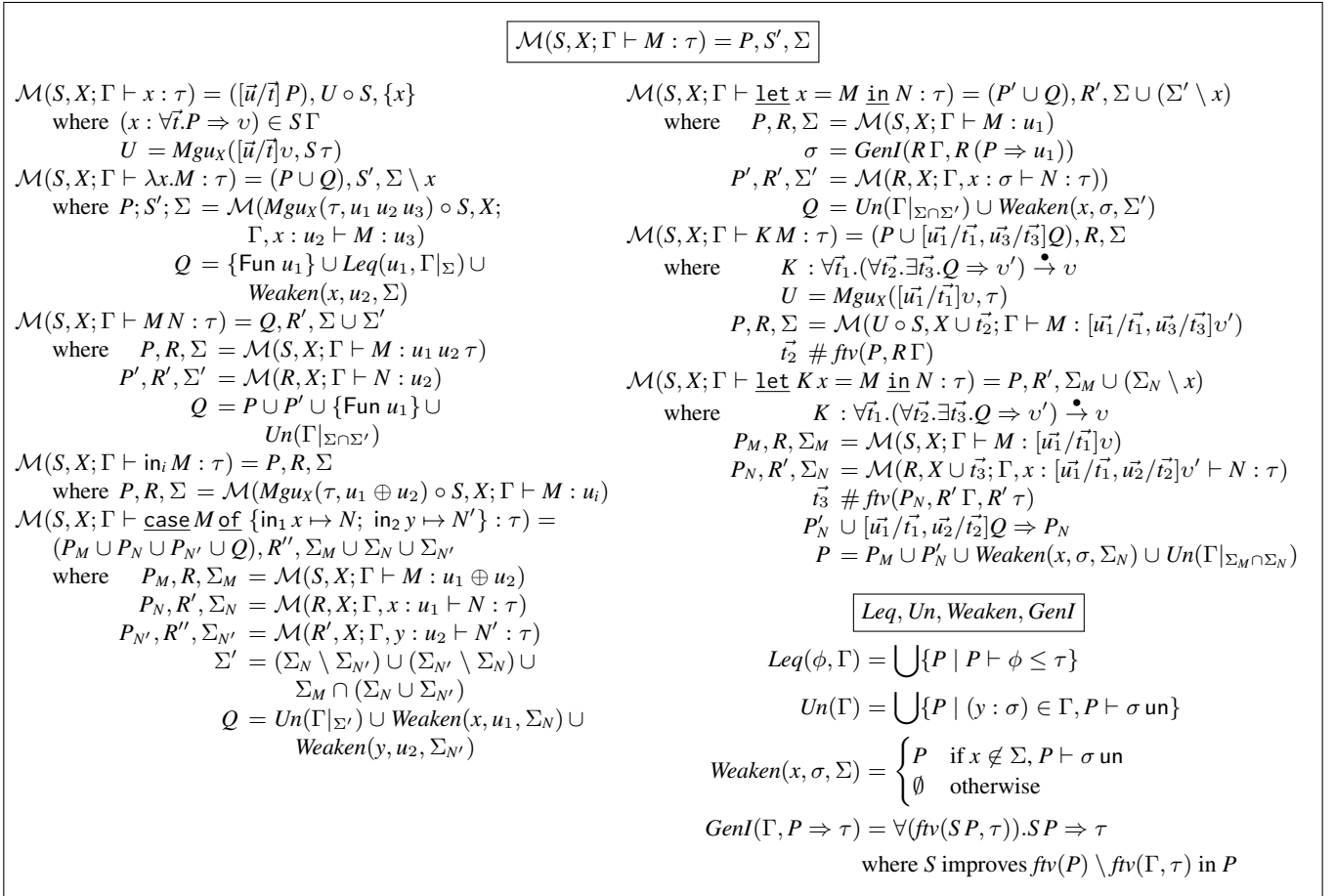

\renewcommand{\arraystretch}{1.1}
\[\fbox{$\Mapp S X \Gamma M \tau = \Mres P {S'} \Sigma$}\]
\[
\begin{array}[t]{@{}ll}
\Wrule{\Mapp S X \Gamma x \tau = \Mres {([\vec u/\vec t]\,P)} {U \circ S} {\Set x}} \\
\Wgiven{
  \Where \multicolumn{2}{@{}l}{(x:\forall \vec{t}. P \then \upsilon) \in S\,\Gamma}\\
  \And U &= \Unif X {[\vec u/\vec t]\upsilon} {S\,\tau}} \\
\Wrule{\Mapp S X \Gamma {\lambda x.M} \tau = \Mres {(P \cup Q)} {S'} {\Sigma \setminus x}}\\
\Wgiven{
  \Where P;S';\Sigma &= \Mapp {\Unif X \tau {u_1\,u_2\,u_3} \circ S} X {\\
  \Break \Gamma,x:u_2} M {u_3} \\
  \And Q &= \Set{\Fun u_1} \cup Leq(u_1,\RestrictEnv{\Gamma}{\Sigma}) \, \cup \\
  \Break Weaken(x,u_2,\Sigma)} \\
\Wrule{\Mapp S X \Gamma {M\,N} \tau = \Mres Q {R'} {\Sigma\cup\Sigma'}}\\
\Wgiven{
  \Where \Mres P R \Sigma &= \Mapp S X \Gamma M {u_1\,u_2\,\tau} \\
  \And \Mres {P'} {R'} {\Sigma'} &= \Mapp R X \Gamma N {u_2} \\
  \And Q &= P \cup P' \cup \Set{\Fun u_1} \, \cup \\
  \Break Un(\RestrictEnv{\Gamma}{\Sigma\cap\Sigma'})} \\
\Wrule{\Mapp S X \Gamma {\ini{M}} \tau = \Mres P R \Sigma} \\
\Wgiven{
  \Where \Mres P R \Sigma &= \Mapp {\Unif X {\tau} {\opt{u_1}{u_2}} \circ S} X \Gamma M {u_i}} \\
\Wrule{\Mapp S X \Gamma {\case{M}{x}{N}{y}{N'}} \tau = } \\
 & \Mres {(P_M \cup P_N \cup P_{N'} \cup Q)} {R''} {\Sigma_M \cup \Sigma_N \cup \Sigma_{N'}} \\
\Wgiven{
  \Where \Mres {P_M} R {\Sigma_M} &= \Mapp S X \Gamma M {\opt{u_1}{u_2}} \\
  \And \Mres {P_N} {R'} {\Sigma_N} &= \Mapp R X {\Gamma,x:u_1} N \tau \\
  \And \Mres {P_{N'}} {R''} {\Sigma_{N'}} &= \Mapp {R'} X {\Gamma,y:u_2} {N'} \tau \\
  \And \Sigma' &= (\Sigma_N \setminus \Sigma_{N'}) \cup (\Sigma_{N'} \setminus \Sigma_{N}) \,\cup \\
  \Break \Sigma_M \cap (\Sigma_N \cup \Sigma_{N'}) \\
  \And Q &= Un(\RestrictEnv{\Gamma}{\Sigma'}) \cup Weaken(x,u_1,\Sigma_N) \,\cup \\
  \Break Weaken(y,u_2,\Sigma_{N'})}
\end{array}
\begin{gathered}[t]
\begin{array}[t]{@{}ll}
\Wrule{\Mapp S X \Gamma {\mlet{x}{M}{N}} \tau = \Mres {(P' \cup Q)} {R'} {\Sigma\cup(\Sigma'\setminus x)}}\\
\Wgiven{
  \Where \Mres P R \Sigma &= \Mapp S X \Gamma M {u_1} \\
  \And \sigma &= GenI(R\,\Gamma,R\,(P \then u_1)) \\
  \And \Mres {P'} {R'} {\Sigma'} &= \Mapp R X {\Gamma,x:\sigma} N \tau) \\
  \And Q &= Un(\RestrictEnv{\Gamma}{\Sigma\cap\Sigma'}) \cup Weaken(x,\sigma,\Sigma')
} \\
\Wrule{\Mapp S X \Gamma {K\,M} \tau = \Mres {(P \cup [\vec{u_1}/\vec{t_1}, \vec{u_3}/\vec{t_3}] Q)} R \Sigma} \\
\Wgiven{
  \Where K &: \forall \vec t_1. (\forall \vec{t_2}. \exists \vec{t_3}. Q \then \upsilon') \uto \upsilon \\
  \And U &= \Unif X {[\vec{u_1}/\vec{t_1}]\upsilon} \tau \\
  \And \Mres P R \Sigma &= \Mapp {U \circ S} {X \cup \vec{t_2}} \Gamma M {[\vec {u_1}/\vec {t_1},\vec {u_3}/\vec {t_3}]\upsilon'}  \\
  \And \vec{t_2} &\# \; ftv(P, R \, \Gamma)
} \\
\Wrule{\Mapp S X \Gamma {\mlet{K\,x}{M}{N}} \tau = \Mres P {R'} {\Sigma_M \cup (\Sigma_N \setminus x)}} \\
\Wgiven{
  \Where K &: \forall \vec t_1. (\forall \vec{t_2}. \exists \vec{t_3}. Q \then \upsilon') \uto \upsilon \\
  \And \Mres {P_M} R {\Sigma_M} &= \Mapp S X \Gamma M {[\vec{u_1}/\vec{t_1}]\upsilon} \\
  \And \Mres {P_N} {R'} {\Sigma_N} &= \Mapp R {X \cup \vec{t_3}} {\Gamma,x:[\vec {u_1}/\vec {t_1}, \vec {u_2}/\vec {t_2}]\upsilon'} N \tau \\
  \And \vec{t_3} &\# \; ftv(P_N, R'\,\Gamma, R'\,\tau) \\
  \And P_N' &\cup \; [\vec{u_1}/\vec{t_1}, \vec{u_2}/\vec{t_2}] Q \then P_N \\
  \And P &= P_M \cup P_N' \cup Weaken(x,\sigma,\Sigma_N) \cup Un(\RestrictEnv{\Gamma}{\Sigma_M \cap \Sigma_N})
}
\end{array} \\[5px]
\fbox{$Leq,Un,Weaken,GenI$} \\ 
\begin{aligned}
Leq(\phi,\Gamma) &= \bigcup \Set{P \mid P \vdash \phi \leq \tau} \\ 
Un(\Gamma) &= \bigcup \Set{ P \mid (y:\sigma) \in \Gamma, P \vdash \Junl \sigma} \\ 
Weaken(x,\sigma,\Sigma) &= \begin{cases}
  P &\text{if $x \not\in \Sigma$, $P \vdash \Junl\sigma$} \\ 
  \emptyset &\text{otherwise}
\end{cases} \\
GenI(\Gamma,P \then \tau) &= \forall (ftv(S\,P,\tau)). S\,P \then \tau \\
& \text{where $S$ improves $ftv(P) \setminus ftv(\Gamma,\tau)$ in $P$}
\end{aligned}
\end{gathered}
\]
\caption{Type inference algorithm $\mathcal{M}$.  We let $u_i$ range over fresh variables, and write $A \# B$ to require that $A$ and $B$ be disjoint.}
\label{fig:infer}
\end{figure*}


The inference algorithm is given in \figref{infer}, in the style of Algorithm~$\M$~\cite{LeeY98}.
The inputs include the environment $\Gamma$, expression $M$ and expected type $\tau$, along with the
current substitution $S$ and the rigid type variables $X$.  The output includes the generated
predicates $P$, the resulting substitution $S'$, and a set of used (term) variables $\Sigma$.  We
let $u_i$ range over fresh type variables, and let $U,R,S$ range over substitutions.  We will look
at illustrative cases of the algorithm in detail; the remaining cases are constructed along the same
lines.

In the variable case, we are given both the variable $x$ and its expected type $\tau$.  We unify
$x$'s actual type, given by $\Gamma$, with its expected type $\tau$.  This illustrates the primary
difference between Algorithm~$\mathcal{M}$ and Milner's Algorithm~$\mathcal{W}$: unification is
moved as close to the leaves as possible.  We defer the details of unification to a separate
algorithm $\Unif X \tau \upsilon$, where the type variables in $X$ are not bound in the resulting
unification procedure.  The implementation of unification does not differ from previous
presentations, such as Jones's unification algorithm for FCP~\cite{Jones97}.  We return any
predicates in the type scheme of $x$, the updated substitution, and the observation that $x$ has
been used.

The application case demonstrates the sets of used variables.  We check the subexpressions $M$ and
$N$; to account for the overloading of functions, we only assume that $M$ has type
$u_1 \, u_2 \, \tau$, for some function type $u_1$.  The variables used in $M$ are captured by
$\Sigma$, and those used in $N$ are captured by $\Sigma'$.  Any variables used in both must be
unrestricted, and so the predicates inferred for the application include not just the predicates
inferred for each sub expression ($P$ and $P'$), but also that any variables used in
$\Sigma \cap \Sigma'$ must have unrestricted type.  We capture this with the auxiliary function
$Un(\restrict{\Gamma}{\Sigma \cap \Sigma'})$, where $\restrict \Gamma \Sigma$ denotes the
restriction of $\Gamma$ to variables in $\Sigma$.  We give a declarative specification of $Un(-)$;
an implementation that finds the simplest such $P$ can be straightforwardly derived from the
definitions of $P \vdash \Junl \cdot$ and entailment.

The \texttt{let} case demonstrates the treatment of polymorphism and binders.
First, we must account for the possibility that $x$ was not used in $N$, and thus must be of
unrestricted type.  This is captured by $Weaken(x,\sigma,\Sigma)$.  Second, we consider
generalization.  Recall the term $(\lambda x.x)\,y$, where $y$ has type $\tau$.  The algorithm will
infer that this term has type $\tau$ under the assumption $\Fun u$ for some variable $u$.  But $u$
appears neither in the typing environment nor in the result, so naively generalizing this expression
would give the (apparently ambiguous) type scheme $(\Fun u) \then \tau$.  However, this is not a
real ambiguity: we have no way of observing the choice of $u$ in the resulting expression, so we
could assume it to be $\lto$ without decreasing the expressiveness or safety of type inference.  We
formalize this observation using an adaptation of Jones's notion of improvement for qualified
types~\cite{Jones95}.  An improving substitution for a qualified type $P \then \tau$ is a
substitution $S$ such that any satisfiable instance of $P \then \tau$ is also a unambiguous
satisfiable instance of $S\,(P \then \tau)$.  For example, $[{\uto}/f]$ is an improving substitution
for $(\Unl f, \Fun f) \then \tau$, as the only ways to prove $\Fun f$ are if $f$ is $\uto$ or $\lto$
and only the former is unrestricted.  In the type $(\Fun f) \then \tau$, where $f$ is not free in
$\tau$ or the environment, we can instantiate $f$ to either $\uto$ or $\lto$ and cannot observe the
choice.  As this choice does not introduce ambiguity, we consider $[{\lto}/f]$ to be an improving
substitution in such cases.  We say that $S$ is an improving substitution for $X$ in $P$ if $S$ is
the union of such improvements for each variable in $X$, and apply such an improving substitution
before generalizing.  Again, we give a declarative specification of $GenI(-,-)$, as the derivation
of its implementation is entirely straightforward.

We can now relate type inference and the syntax-directed type system.  First, inference constructs
valid typings.
\begin{theorem}[Soundness of $\M$]\label{thm:soundness-m}
  If $\Mapp S X \Gamma M \tau = \Mres P {S'} \Sigma$, then
  $S'\,P \mid S'\,(\restrict \Gamma \Sigma) \vdash M: S'\,\tau$.
\end{theorem}
\noindent%
The proof is by induction on the structure of $M$; each case involves comparing the predicates
generated in inference to the predicates needed for typing.  In combination with
\thmref{soundness-s}, this gives a similar soundness result for inference with respect to the
original type system.  Next, we want to show that any valid typing can be found by inference.
\begin{theorem}[Completeness of $\M$]\label{thm:completeness-m}
  If $S$ is a substitution and $X$ is a set of type variables such that
  $P \mid S\,\Gamma \vdashS M : S\,\tau$, and $\restrict S X = id$, then
  $\Mapp {id} X \Gamma M \tau = \Mres Q {S'} \Sigma$ such that
  $(P \then S\,\tau) \lessgeneral GenI(S'\,\Gamma, S'\,Q \then S'\,\tau)$.
\end{theorem}
\noindent%
The proof is by induction on the typing derivation, observing in each case that the computed type
generalizes the type in the derivation.  Again, in combination with \thmref{completeness-s}, we have
a completeness result for inference with respect to the original type system.  Finally, this allows
us to give a constructive proof that \lang enjoys principal types.
\begin{theorem}[Principal Types]\label{thm:principal}
  If $P_0 \mid \Eta \vdash M: \sigma_0$ and $P_1 \mid \Eta \vdash M: \sigma_1$ then there is
  some $\sigma$ such that $\emptyset \mid \Eta \vdash M: \sigma$ and $(P_0 \mid \sigma_0) \lessgeneral
  \sigma, (P_1 \mid \sigma_1) \lessgeneral \sigma.$
\end{theorem}
\noindent%
The soundness of inference tells us that, if there are any typings for a term in an environment,
then the inference algorithm will compute some typing for that term.  The completeness of inference
tells us that the computed type will be at least as general as the original types.

\subsection{Conservativity of Typing}\label{sec:conservativity}
\newcommand{\vdashSO}{\vdashS_{\raisebox{0px}{{\tiny\upshape\!\!\!OML}}}}

We have claimed that \lang is as expressive as functional languages without linearity.  To
formalize that claim, we will show that any expression typeable in OML, Jones's core calculus for
qualified types~\cite{Jones94}, is also typeable in \lang.

OML is a Core ML-like language with qualified types.  Its types and terms are pleasingly simple: the
former contains functions, type variables, and qualified and quantified types, and the latter
contains variables, applications, abstractions, and \texttt{let} (to introduce polymorphism).  We do
not give a full description of OML typing here, partly as it is so similar to \lang typing.  In
particular, as in \lang typing, OML has a syntax directed typing judgment
$P \mid \Gamma \vdashSO M : \tau$, where $P$ is a collection of predicates, $\Gamma$ an OML typing
environment and $\tau$ an OML type.

The crux of our argument is that (by construction) the syntax-directed typing rules of \lang can
each be seen as generalizations of the corresponding rules of OML.  For example, rules
\strule{\I\to} and \strule{\E\to} can introduce generalization over function types, a feature not
present in OML.  However, they need not do so; if all functions are unrestricted, \strule{\I\to} and
\strule{\E\to} are elaborate restatements of the corresponding rules of OML.  The remaining
difference is in the treatment of variables: \lang may insist on predicates to capture their
unrestricted use, where there are no corresponding predicates required by OML, but this will never
cause a term to be ill-typed.
\begin{theorem}\label{thm:conservative}
  If $P \mid \Gamma \vdashSO M : \tau$, then there is some $Q$ such that $Q \mid \Gamma \vdashS
  M : \tau$, and $Q \entails P$.
\end{theorem}
\noindent%
OML also has a sound and complete type inference algorithm, and principal types.  Thus, we see that
if OML type inference accepts a given term, then \lang type inference will also accept the term, and
in each case will compute its most general typing.  We might hope to show the converse as well;
however, we do not know of a non-linear core calculus that matches the exact features of \lang,
including both qualified types and data type-mediated first class polymorphism.

We do not suggest that terms are given the same types in each setting: for example, the function
$\lambda x. \lambda y.y$ is given the type $\forall t u. t \to u \to u$ in OML, whereas it would
given the type $\forall t u. \Un t \then t \to u \to u$ in \lang.  Similarly, sums must be taken as
primitive in \lang (as their elimination form shares its environment), whereas they can be encoded
in OML.  Finally, as demonstrated in our earlier discussion of monads~\secref{inaction}, \lang may
suggest more refined abstractions than are present in non-linear languages.  Nevertheless, this
result does show a strong connection between programming in \lang and programming in traditional
functional languages, one which is not shared by other combinations of linear and functional
programming.

\section{Semantics}\label{sec:semantics}
\newcommand{\LinV}{LinVals}

We motivated the discussion of linear type systems by considering examples like session types and
mutable arrays, in which we wanted to avoid duplicating or discarding values of linear types.  The
\lang type system, however, only restricts the use of assumptions, and says nothing about the use of
values directly.  Further, \lang differs from other substructural calculi in several ways, including
the use of overloading and the form of first-class polymorphism.  In this section, we demonstrate
that \lang assures that the use of values, not just of assumptions, is consistent with their typing.
To do so, we define a natural semantics for \lang terms, annotated with the values introduced and
eliminated in the course of evaluation.  We can then show that any values used non-linearly have
unrestricted type.  Our approach is strongly inspired by that used by Mazurak et
al. to prove a similar property of their \Fpop calculus~\cite{MazurakZZ10}.

We begin by defining a notion of values for \lang.  Intuitively, we might expect values to be
abstractions, sums of values, or constructors applied to values.  However, our intended safety
property requires that we distinguish different instances of syntactically-identical values.  The
top of \figref{semantics} gives an extended syntax of \lang, in which values are tagged with indices
from some index set $Ix$.  The semantics will tag new values with fresh indices, and we will then
rely on the indices to establish identity when showing safety.  Given a set of assumptions $P$, we
identify a subset of values, $\LinV_P$, as linear:
\[
  \LinV_P = \Set{V \in Value \mid \text{if $\vdash V : \tau$, then $P \nentails \Un \tau$}}.
\]
Our goal is to show that values in $\LinV_P$ are neither duplicated nor discarded during evaluation.
The contents of $\LinV_P$ depend on the signatures of the constructors.  For a simple example,
suppose that we have some $K$ with signature $(\exists u. u) \uto T$.  To show $P \entails \Un T$,
we would have to show that $P \entails \Un u$ (where $u \not\in ftv(P)$).  This is clearly
impossible, so $K\,V \in \LinV_P$ for any value $V$ and non-trivial $P$.  On the other hand
$\lambda x. x$ is not in $\LinV_P$, as it can be given unrestricted type.

\begin{figure}
\begin{syntax}
  & Ix \ni j,k  \\
  & Value \ni V,W & ::= & K^j \, V \mid \lami j x M \mid \inii j V \\
  & M,N & ::= & V \mid \dots
\end{syntax}
\hfil\rule{0.8\linewidth}{.3pt}
\vspace{-1em}
\begin{gather*}
\infbox{\irule{\text{$j$ fresh}};{\red{\lambda x.M}{\lami j x M}{\lami j x M}{\emptyset}}}
\isp
\infbox{\irule{\red{M}{V}{I}{E}}
              {\text{$j$ fresh}};
              {\red{K \, M}{K^j \, V}{I,K^j\,V}{E}}}
\\
\infbox{\irule{\red{M}{\lami j x {M'}}{I}{E}}
              {\red{N}{V}{I'}{E'}}
              {\red{[V/x]M'}{W}{I''}{E''}};
              {\red{M\,N}{W}{I,I',I''}{E,E',E'',\lami j x {M'}}}}
\\
\infbox{\irule{\red{M}{V}{I}{E}}
              {\text{$j$ fresh}};
              {\red{\ini{M}}{\inii{j}{V}}{I,\inii{j}{V}}{E}}}
\isp
\infbox{\irule{\red{M}{V}{I}{E}}
              {\red{[V/x]N}{W}{I'}{E'}};
              {\red{\mlet x M N}{W}{I,I'}{E,E'}}}
\\
\infbox{\irule{\red{M}{\inii{j}{V}}{I}{E}}
              {\red{[V/x_i]N_i}{W}{I'}{E'}};
              {\red{\case{M}{x_1}{N_1}{x_2}{N_2}}{W}{I,I'}{E,E',\inii{j}{v}}}}
\end{gather*}
\caption{Linearity-aware semantics for \lang.}\label{fig:semantics}
\end{figure}

The bottom of \figref{semantics} gives a natural semantics for \lang.  The evaluation relation
$\red M V I E$ denotes that $M$ evaluates to $V$; the annotations $I$ and $E$ are multisets of
values, $I$ capturing all values introduced during the evaluation and $E$ capturing all values
eliminated during the evaluation.  (We track values, rather than just indices, so that we can state
the type safety theorem below.)  Functions evaluate to themselves, but annotated with a fresh index.
The only value introduced is the function, and no values are eliminated.  The other introduction
forms are similar, but must account for the evaluation of their subexpressions.  We use
call-by-value evaluation; as observed by Mazurak et al.~\cite{MazurakZZ10}, call-by-name and
call-by-need evaluation may result in discarding linearly typed values during evaluation.  The
values introduced in evaluating an application are those introduced in evaluating each of its
subexpressions and in evaluating the substituted result of the application.  The values eliminated
are those eliminated in each hypothesis and the function itself.  The \texttt{let} and \texttt{case}
rules are similar.

We can now state our desired safety property.  Intuitively, if $\red M V I E$, we expect that each
linear value introduced during evaluation (that is, each $W \in I \cap \LinV_P$) will appear either
exactly once, either in $E$ or as a subexpression of the result $V$.  For any expression $M$, we
define $SExp(M)$ to be the subexpressions of $M$, defined in the predictable fashion,
$Exp(M) = \Bag M \cup SExp(M)$, and $Val(M) = Exp(M) \cap Value$.
\begin{theorem}[Type safety]\label{thm:safety}
  Let $M$ be a closed term such that $\vdash M: \forall t. P \then \tau$ and $\red M V I E$.
  \begin{enumerate}
  \item $P \mid \emptyset \vdash V : \tau$.
  \item Let $E' = E \cup Val(V)$, and let $D = I \setminus E'$ (the values discarded during
    evaluation) and $C = E' \setminus I$ (the values copied during evaluation). Then,
    $W \in D \cup C$ only if $W \not\in \LinV_P$.
  \end{enumerate}
\end{theorem}
\noindent%
The proof is by induction over the structure of $M$.  The key observation is that duplication and
discarding can happen only as the result of substitution, and thus that linearity of variables
(i.e., assumptions) is enough to assure that only unrestricted values are duplicated.  We believe
this argument can be straightforwardly generalized to small-step semantics, again following Mazurak
et al.~\cite{MazurakZZ10}

\section{Extensions}\label{sec:extensions}
\newcommand{\Dup}[1]{\predh{\mathtt{Dup}}{#1}}
\newcommand{\Drop}[1]{\predh{\mathtt{Drop}}{#1}}

We describe three extensions of \lang, showing the generality and flexibility of our approach.

\lang has a linear type system, in which both contraction (duplication) and weakening (discard) are
limited to unrestricted types.  Several alternative substructural logics exist: relevant logics, for
example, exclude weakening but not contraction, and affine logics exclude contraction but not
weakening.  Some systems, such as that of Ahmed et al.~\cite{AhmedFM05} and Gan et
al.~\cite{GanTM14} provide linear, affine, relevant, and unrestricted types simultaneously.
Finally, there have been several type systems that introduce similar partitioning of assumptions to
control side-effects, starting from Reynolds' work on Idealized Algol~\cite{Reynolds78} and
continuing with modern work on bunched implication~\cite{OHearnP99} and separation logic.  We have
focused on the linear case in particular because various examples, such as session types, require
its restrictions on both contraction and weakening.  Nevertheless, we believe the \lang approach
would apply equally well in these other cases.  For example, the type system we have given has a
single predicate, $\Un \tau$, used both when assumptions are duplicated and when they are discarded.
Alternatively, we could introduce distinct predicates for these cases, say $\Dup\tau$ and
$\Drop\tau$.  We could then redefine our existing predicate $\Un\tau$ as the conjunction of
$\Dup\tau$ and $\Drop\tau$.  As in the systems of Ahmed at al. and Gan et al., we would require four
arrow types.  However, the remainder of the \lang approach would adapt seamlessly.  We could extend
the $\Fun\tau$ predicate and the $\geq$ relation to include the new arrow types, and the resulting
system would continue to enjoy principal types and complete type inference.

The treatment of functions differs from the other primitive types (like products and sums) because
the linearity of a function from $\tau$ to $\upsilon$ cannot be determined from linearity of $\tau$
and $\upsilon$.  A similar observation can be made of existential types.  For example, suppose that
we have two constructors with signatures $K_1 :: (\exists u. u) \uto T_1 \, t$ and
$K_2 :: (\exists u. \Un u \then u) \uto T_2 \, t$.  Assumptions of type $T_1\,\tau$ will always be
treated as linear, and assumptions of type $T_2\,\tau$ always unrestricted, regardless of the choice
of $\tau$.  We could view $T_1$ and $T_2$ as instances of a general $T$ pattern (i.e., as the
satisfying instances of a predicate $\predh{T}{t}$) just as we view $\uto$ and $\lto$ as instances
of a general $\to$ pattern.  Following the approach taken for functions, we would extend the $\geq$
relation to include the $T$ types, asserting that $\cdot \entails \tau \geq T_1\,\upsilon$ and
$\Un \tau \entails \tau \geq T_2\,\upsilon$.  We would then introduce a generalized constructor
$K :: (\predh{T}{t}, u \geq t) \then u \uto t$, and a generalized deconstructor
$unK :: (\predh{T}{t}, t \geq f) \then t \to (\forall u. u \to r) \labto f r$.  As the goal is
generalizing over the predicate $\Un u$, the body of the deconstructor cannot rely on its presence.
This example demonstrates the flexibility of \lang; in particular, it shows that our treatment of
functions is an instance of a more general pattern, itself expressible in \lang.  We suspect that
this treatment of existentials would also come closest to capturing how the use of existentials in
non-linear functional languages could be expressed in linear calculi.  On the other hand, there are
cases for which this approach would not be appropriate, such as the use of existentials to enforce
linear use of unrestricted primitives.  We believe that more practical experience would be require
to determine how, and how often, this generalization of existentials should be applied.

We have treated the \texttt{dup} and \texttt{drop} methods as providing a helpful intuition for the
use of the \texttt{Un} predicate, but have assumed that their explicit use or implementation is not
of interest.  However, there are cases in which providing non-trivial implementations of these
methods could be useful. For example, many operating system resources, such as file handles, need to
be explicitly freed.  One could imagine capturing such resources as affine types in a language based
on \lang, in which the \texttt{drop} method freed the underlying resource.  Similarly, given
suitable primitives, one could imagine using \texttt{drop} and \texttt{dup} to implement a kind of
reference-counting scheme for resources, in which \texttt{dup} incremented the reference count and
\texttt{drop} decremented it.  This approach would generalize the various scope based mechanisms for
managing such resources in languages such as C\# and Java.  The derived definitions of \texttt{drop}
and \texttt{dup} for products and sums~\secref{class} would extend to this setting as well.
However, this would introduce a new concern: the placement of automatic inserted calls to the
\texttt{drop} and \texttt{dup} methods.  For a simple example, imagine that some variable $y$ is in
scope in the expression $\lambda x.M$, but not free in $M$.  We must insert a call to
$\mathtt{drop}\,y$.  Our current syntax-directed approach could be interpreted as moving calls to
\texttt{drop} to the leaves of the typing derivation, but in this case that would delay the discard
of $y$ until the function $\lambda x.M$ is invoked, which might be undesirable.


\section{Related Work}\label{sec:related}

The past thirty years have seen a wealth of work on linear types and their applications.  We
summarize some of the work most directly related to our own.

In introducing substructural type systems~\secref{other-linear}, we described several other general
purpose calculi, including \Fpop of Mazurak et al.~\cite{MazurakZZ10}, and Alms of Tov and
Pucella~\cite{TovP11}.  These systems were both influential on the development of \lang.  Our work
differs from theirs in two regards.  First, we have generalized the treatment of functions, and thus
increased the expressiveness of function combinators.  We believe that, especially given the
importance of combinator-based idioms in functional programming, this is a significant advance in
the usability of linear functional calculi.  Second, our treatment relies on qualified types, rather
than building notions of subkinding, subtyping, and variance into the type system itself.  While
this may seem to simply be trading one kind of complexity for another, we believe that qualified
types are an independently useful language feature (a claim borne out by the experience of Haskell).
Finally, we believe that qualified types are a natural way to express relationships among types, as
demonstrated by our generalization of relative linearity to encompass existential types.

\Fpop, Alms, and \lang all rely on identifying a collection of types as unrestricted (through kind
mechanisms in the first cases and type predicates in ours).  There are several other mechanisms to
integrate linear and unrestricted types.  Wadler~\cite{Wadler93} and Barber and
Plotkin~\cite{BarberP96} give calculi based directly on the propositions and proofs of linear logic,
in which each linear type $\tau$ has an intuitionistic counterpart $!\tau$.  These calculi also draw
a distinction between intuitionistic and linear assumptions, where only the former are subject to
contraction and weakening.  While these calculi have close logical connections, the manipulation of
the $!$ modality adds significant syntactic bureaucracy, and does not provide an obvious route to
generalizing linear and unrestricted behavior.  Walker~\cite{Walker04} and Ahmed et
al.~\cite{AhmedFM05} present systems of annotations for linearity (albeit without polymorphism).
These approaches seem less well suited for programming with linear types, however.  For example,
they provide linear Booleans (of little expressive value, as the duplication and discarding
operations for Booleans can be easily defined) and unlimited session-typed channels (presumably an
empty type).  Finally, they require all types to be annotated with linearity (or usage) annotations,
which is acceptable in a core language but unsuited to languages used by humans. Clean adopts a
similar annotation-based approach in its uniqueness typing system~\cite{SmetsersBEP93}.  However,
the aims of uniqueness typing and linearity are dual: in Clean, unique values can become non-unique
(at the cost of some of their operations), while in a linear type system we must guarantee linearity
(but can use unrestricted values linearly).

Gustavsson and Svenningsson~\cite{GustavssonS00} describe a system of usage annotations and bounded
usage polymorphism; Hage et al~\cite{HageHM07} describe an alternative approach to usage inference
based on effect typing and subeffects.  These approaches differ from linear type systems in two
ways.  First, they treat usage separately from types; as we argued in the last paragraph, this
produces confusing or empty types, like linear Booleans or unrestricted channels.  Second, linearity
is prescriptive, while usage types are descriptive.  This means that usage can be approximated,
where approximations of linearity would either lose safety or expressiveness.  Consequently, usage
analyses can be invisible to the programmer, whereas linearity must (to some degree) be
programmer-visible.  Nevertheless, our $\leq$ predicate seems similar to approaches to
subeffecting for usage and strictness~\cite{HoldermansH10}, and we believe that investigating this
similarity is valuable future work.

Finally, there have been numerous substructural approaches to typing for imperative and low-level
languages, including region types~\cite{WalkerCM00}, alias types~\cite{SmithWM00}, adoption and
focus~\cite{FahndrichD02}, and linear types for locations~\cite{MorrisettAF05}, and several
generalizations of linear typing, including coeffect systems~\cite{PetricekOM14}.  These approaches
have similar goals to our work---establishing safety guarantees beyond those expressed in
traditional type systems---but differ in their underlying calculi and do not share our focus on
principality and type inference.  Nevertheless, some of the ideas of these systems could be
profitably applied in ours.  For example, some adaptation of the adoption and focus mechanisms could
avoid the rebinding present in cases such as our dyadic session types example. We think exploring
the overlap of our system and the problems they address, such as exploring explicit memory
management in a \lang-like language, will be important future work.

\section{Future Work}\label{sec:future}

We have presented \lang, a new linear functional language achieving both the safety guarantees made
possible by linear types and the expressiveness of conventional functional programming languages.
We have demonstrated several examples of linear and functional programming in \lang.  We have shown
that \lang has principal types and decidable type inference, that it is a conservative extension of
existing functional calculi, and that reduction preserves linearity.  We have also shown several
simple extension of the core \lang calculus, incorporating more flexible treatment of existentials,
and other notions of substructural typing.  We conclude by discussing several directions for future
work.

We intend \lang to provide a foundation for practical functional programming with linear types.
This can be tested in two ways.  First, we intend to explore abstractions for linear programming.
We hope to draw on existing mechanisms, such as adoption and focus~\cite{FahndrichD02} in the
imperative setting and parameterized monads~\cite{Atkey09} in the functional setting, while taking
advantage of \lang's first-class treatment of linearity to express these mechanisms within the
language.  Second, we hope to build larger programs in \lang, taking advantage of linearity to
enhance safety properties; domains like concurrency and low-level programming seem particularly
suited to such an approach.

\lang distinguishes between linear and unrestricted functions for type safety reasons, while the
(high-level) semantics we give treats abstractions identically.  Similar distinctions are drawn by
usage type systems, for efficiency reasons.  We believe that similar efficiency gains could be
obtained in compiling \lang programs.  As our treatment of functions is general, similar approaches
could be applied to other types for efficiency reasons as well.  For example, rather than requiring
that arrays be treated linearly, we could overload the array operations to apply to both linear and
unrestricted arrays, but to use efficient in-place operations when arrays were used linearly.

The central technical problem addressed by \lang is the multiplication of function spaces.  Linear
type systems are not the only context in which this can occur.  Similar multiplications happen, for
example, in type and effect systems or in systems that distinguish pointed and unpointed
types~\cite{LaunchburyP96}.  We believe that the approach taken in \lang would generalize to such
cases as well; in particular, we believe that tracking pointedness could be relevant in many of the
same application domains in which linearity is relevant.

\section*{Acknowledgments}

I thank James Cheney for feedback on drafts of this paper, and James McKinna, Sam Lindley, and my
other colleagues for helpful discussions. This work was funded by EPSRC grant number EP/K034413/1.

\bibliographystyle{abbrvnat}
\bibliography{subclass}{}

\flushcolsend


\clearpage
\appendix

\section{Packaging Unrestricted Channels}\label{sec:packaging-channels}

We might want to express session typing by wrapping an underlying unrestricted implementation of
untyped channels (which we will call \texttt{Chan}, patterned on the Haskell \texttt{Chan} type).
Previous work has demonstrated the use of existential types in doing this kind of wrapping.  We have
two problems:
\begin{enumerate}
\item How to capture the linearity of sessions, while still allowing $\gvend$ channels to be
  unrestricted; and,
\item How to capture the types of sent and received values.
\end{enumerate}
Ideally, we would like a solution that accomplishes both using existential types, avoiding the need
for any waffle about the module system.  Solving the first is actually relatively easy, and just
relies on existing classes and simple existential types:
\begin{code}
instance Un End
data Ch s = c >= s =>
            PackCh c (Dynamic -> c -> M c)
                     (c -> M (Dynamic, c))

makeChannel :: Chan Dynamic -> Ch s
makeChannel c =
  PackCh c (\v c -> do writeChan c v; return c)
           (\c -> do v <- readChan c; return (v, c))

send :: t >= f => t -> Ch (t :!: s) -f> M (Ch s)
send v (Ch c sender receiver) =
  do c <- sender (toDyn v) c
     return (Ch c sender receiver)

receive :: Ch (t :?: s) -> M (t, Ch s)
receive (Ch c sender receiver) =
  do (v, c) <- receiver c
     return (fromDyn undefined v, Ch c sender receiver)
\end{code}
However, while this enforces linearity, it does not guarantee session typing.  In particular, code
with access to the \texttt{Ch} type may send a \texttt{Dynamic} value containing the wrong type,
causing the corresponding \texttt{fromDyn} to fail.  We can do better if we assume a notion of type
equality (at the cost, of course, of significant additional complexity in the type system).  We
start with type equality, which we can define using functional dependencies:
\begin{code}
class t ~~ u | t -> u, u -> t
instance t ~~ t
\end{code}
We can then define the channel type as follows.
\begin{code}
instance Un End
data PChan s = PChan (Chan Dynamic)
data Ch s = c s >= s =>
            PackCh (c s)
                   ((t >= f, s ~~ (t :!: s')) =>
                    t -> c s -> M (c s'))
                   ((s ~~ (t :?: s')) =>
                    c s -> M (t, c s'))
\end{code}
The type \texttt{PChan} wraps an unrestricted channel with a phantom type variable.  The type
\texttt{Ch} follows the same pattern as before, but now encoding the form of the \texttt{send} and
\texttt{receive} functions in the packaged sender and receiver.  Consequently, this version depends
on both first-class existentials and universals.  Correspondingly, the implementations move the
introduction and elimination of the \texttt{Dynamic} type into the packaged functions, but are
otherwise unchanged.
\begin{code}
makeChannel :: Chan Dynamic -> Ch s
makeChannel c =
  PackCh (PChan c)
         (\v (PChan c) -> do writeChan c (toDyn v)
                             return (PChan c))
         (\(PChan c) -> do v <- readChan c
                           return (fromDyn undefined v,
                                   PChan c))

send :: t >= f => t -> Ch (t :!: s) -f> M (Ch s)
send v (Ch c sender receiver) =
  do c <- sender v c
     return (Ch c sender receiver)

receive :: Ch (t :?: s) -> M (t, Ch s)
receive (Ch c sender receiver) =
  do (v, c) <- receiver c
     return (v, Ch c sender receiver)
\end{code}

\section{Encoding Products in \lang}
\newcommand{\enc}[1]{\mathcal{E}(#1)}

We define an encoding $\enc{-}$ from an extension of Quill with multiplicative and additive products
to Quill without products.  To review, the terms and typing rules for the additive product are as
follows.
\[
\infbox{\irule{\Gamma \vdash M : \tau}
              {\Gamma \vdash N : \upsilon};
              {\Gamma \vdash [M, N] : \tau \with \upsilon}}
\quad
\infbox{\irule{\Gamma \vdash M : \tau \with \upsilon};
              {\Gamma \vdash \key{fst}\, M : \tau}}
\quad
\infbox{\irule{\Gamma \vdash M : \tau \with \upsilon};
              {\Gamma \vdash \key{snd} \, N : \upsilon}}
\]
The terms and typing rules for the multiplicative product are as follows.
\begin{gather*}
\infbox{\irule{\Gamma \vdash M: \tau}
              {\Gamma' \vdash N: \upsilon};
              {\Gamma, \Gamma' \vdash (M, N) : \tau \otimes \upsilon}}
\\
\infbox{\irule{\Gamma \vdash M : \tau \otimes \tau'}
              {\Gamma', x : \tau, y : \tau' \vdash N : \upsilon};
              {\Gamma \vdash \mlet{(x,y)}{M}{N} : \upsilon}}
\end{gather*}

The encoding of multiplicative products is simply the typical Church encoding of products.  We
assume the following constructor to capture the use of universal types:
\[
  MP :: \forall t u. (\forall v f. (t \moreunlimited f, u \moreunlimited f) \then (t \to u \to v) \labto{f} v) \uto MP \, t \, u.
\]
We can then define the encoding of the multiplicative product and its terms.
\begin{align*}
  \enc{\tau\otimes\upsilon} &= MP \, \enc\tau \, \enc\upsilon \\
  \enc{(M, N)} &= MP\,(\lambda f. \, f\,\enc M\,\enc N) \\
  \enc{\mlet{(x,y)} M N} &= \begin{array}[t]{@{}l@{}}
                              \key{let}\;MP\,f = \enc{M} \\
                              \key{in}\; f\, (\lambda x. \lambda y. \enc N)
                            \end{array}
\end{align*}

The encoding of additive products has a similar flavor, but must be defined in terms of (additive)
sums.  Again, we assume a constructor for a suitable universal type:
\begin{multline*}
  AP :: \forall t u. (\forall v f. (t \moreunlimited f, u \moreunlimited f) \then \\ ((t \labto{f} v) \oplus (u \labto{f} v) \labto{f} v)) \uto AP \, t \, u.
\end{multline*}
We can then encode the additive product and its terms as follows.
\begin{align*}
  \enc{\tau \with \upsilon} &= AP\,\enc{\tau}\,\enc{\upsilon} \\
  \enc{[M, N]} &= AP (\lambda l.\, \key{case}\;l\;\key{of}\;\{\;
                  \begin{array}[t]{@{}l@{}}
                    \inl{f} \mapsto f\,\enc{M}; \\
                    \inr{f} \mapsto f\,\enc{N} \}
                  \end{array} \\
  \enc{\key{fst}\, M} &= \mlet{AP\,f}{\enc{M}}{\inl{id}} \\
  \enc{\key{snd}\, M} &= \mlet{AP\,f}{\enc{M}}{\inr{id}}
\end{align*}


\section{Proofs}\label{sec:proofs}
\newenvironment{fake}[1]{\par\vspace{3pt}\noindent\textbf{#1}\itshape}{\normalfont\ignorespacesafterend\vspace{3pt}\par}

\subsection{Syntax-Directed Type System}

\begin{fake}{Theorem~\ref{thm:soundness-s} \textnormal{(Soundness of $\vdashS$)}.}
  If $P \mid \Gamma \vdashS M: \tau$ and $P \vdash \Eta \approx \Gamma$, then $P \mid \Eta \vdash M:
  \tau$.
\end{fake}

\begin{proof}
By structural induction on the derivation of $P \mid \Gamma \vdashS M: \tau$.
\begin{itemize}
\item \textit{Case \strule{var}.} We have a derivation of $P \mid x : \sigma \vdash x : \sigma$ by
  \trule{var}.  We construct the necessary derivation in three further steps.  First, as $ Q \then
  \tau \lessgeneral \sigma$, we can construct a derivation of $P \mid x:\sigma \vdash x: Q \then
  \tau$ by repeated applications of \trule{$\E\forall$}.  Second, as $P \entails Q$, we can
  construct a derivation of $P \mid x:\sigma \vdash x : \tau$ by repeated application of
  \trule{\E\then}.  Finally, as $P \vdash \Junl\Gamma$, we can construct a derivation of $P \mid
  \Gamma,x:\sigma \vdash x : \tau$ by using \trule{\W} for each binding in $\Gamma$.
\item \textit{Case \strule{\I\to}.} By \theIH, we have a derivation of $P \mid \Eta,x:\tau
  \vdash M: \upsilon$.  Apply \trule{\I\to}, and reusing the derivations of $\Fun \phi$ and $\phi
  \lessgeneral \Gamma$, we construct a derivation of $P \mid \Gamma \vdash \lambda x.M: \phi \tau
  \upsilon$.
\item \textit{Case \strule{\E\to}.} By \theIH, we have derivations of $P \mid \Gamma,\Delta \vdash
  M: \tau \to \upsilon$ and $P \mid \Gamma,\Delta' \vdash N: \tau$.  Applying \trule{\E\to},
  reusing the derivation of $P \vdash \Fun \phi$, we construct a derivation of $P \mid
  \Gamma,\Gamma,\Delta,\Delta' \vdash M\,N : \upsilon$.  Finally, as $P \vdash \Junl\Gamma$, we can
  apply \trule{\C} for each binding in $\Gamma$, constructing a derivation of $P \mid
  \Gamma,\Delta,\Delta' \vdash M \, N : \upsilon$.
\item \textit{Case \strule{\I\oplus$_i$}} is direct from \theIH.
\item \textit{Case \strule{\E\oplus}.} From \theIH we have a derivation of $P \mid
  \Gamma,\Gamma',\Delta, \Delta \vdash \case{M}{x}{N}{x}{N'} : \upsilon$.  We can then repeated apply
  rule \trule{Ctr}, justified by the assumption $P \vdash \Junl\Delta$, to derive $P \mid
  \Gamma,\Gamma',\Delta \vdash \case{M}{x}{N}{x}{N'} : \upsilon$.
\item \textit{Cases \strule{make} and \strule{break}} follow immediately from \theIH.
\item \textit{Case \strule{let}.} By \theIH, we have a derivation
  $Q \mid \Gamma,\Delta,x:\tau \vdash M: \tau$.  We can construct a derivation of
  $\emptyset \mid \Gamma,\Delta,x:\sigma \vdash M: \sigma$ by application of \trule{\E\forall} and
  \trule{\E\then} at each use of variable $x$ and application of \trule{\I\forall} and
  \trule{\I\then} at the conclusion of the derivation.  We also have a derivation of
  $P \mid \Gamma,\Delta',x:\sigma \vdash N: \upsilon$ by \theIH.  Applying \trule{let} gives a
  derivation of $P \mid \Gamma,\Gamma,\Delta,\Delta' \vdash \mlet{x}{M}{N} : \upsilon$.  Finally, as
  $P \vdash \Junl\Gamma$, we can apply \trule{\C} repeatedly to derive
  $P \mid \Gamma,\Delta,\Delta' \vdash \mlet{x}{M}{N}: \upsilon$.  \qedhere
\end{itemize}
\end{proof}

\begin{fake}{Theorem \ref{thm:completeness-s} \textnormal{(Completeness of $\vdashS$)}.}
  If $P \mid \Eta \vdash M: \sigma$ and $P \vdash \Eta \approx \Gamma$, then there are some $Q$ and
  $\tau$ such that $Q \mid \Gamma \vdashS M: \tau$ and
  $(P \mid \sigma) \lessgeneral Gen(\Gamma,Q \then \tau)$,
\end{fake}

We begin with helpful intermediate results.

\begin{lemma}\label{thm:more-specific}
  If $P \mid \Gamma \vdashS M: \tau$, then, letting $\sigma = Gen(\Gamma, P \then \tau)$, for any
  $P' \then \tau' \lessgeneral \sigma$, $P' \mid \Gamma \vdashS M: \tau'$.
\end{lemma}

\begin{proof}
  Let $\sigma = \forall \vec t. Q \then \upsilon$.  By definition, there are some $\vec\phi$ such
  that $\tau' = [\vec\phi / \vec t]\upsilon$ and $P' \entails [\vec\phi / \vec t]Q.$ Therefore, we
  have $P' \mid [\vec\phi / \vec t]\Gamma \vdash M: \tau'$.  Finally, since the $\vec t$ are free in
  $\Gamma$, we have that $P' \mid \Gamma \vdash M: \tau'$.
\end{proof}

Define $\Eta \lessgeneral \Eta'$ if $\dom(\Eta) = \dom(\Eta')$ and for each $x \in
\dom(\Eta)$, $\Eta(x) \lessgeneral \Eta'(x)$.  Define $\Delta \lessgeneral \Delta'$
similarly.

\begin{lemma}\label{thm:more-general-context}
  If $P \mid \Delta \vdashS M: \tau$, and $\Delta \lessgeneral \Delta'$, then $P \mid \Delta' \vdashS M:
  \tau.$
\end{lemma}

\begin{proof}
  The proof is by induction on the derivation of $P \mid \Delta \vdash M: \tau$; the only
  interesting case is for \strule{var}, which depends on the observation that if $(P \then \tau)
  \lessgeneral \sigma$ and $\sigma \lessgeneral \sigma'$ then $(P \then \tau) \lessgeneral \sigma'$.
\end{proof}

\begin{lemma}\label{thm:weakening}
  If $P \mid \Eta \vdash M: \tau$, $y$ is not free in $M$, and $P \vdash \Junl\sigma$, then $P
  \mid \Eta, y:\sigma \vdash M: \tau$.
\end{lemma}

\begin{proof}
  By induction on the derivation of $P \mid \Eta \vdash M: \tau$.  In the \strule{Var} case, we know
  that $M$ is an expression $x$, $\Eta$ is $\Eta', x : \sigma'$ for some $\Eta'$ such that
  $P \vdash \Junl{\Eta'}$, and, as $y$ is not free in $M$, $y$ is not $x$.  Given that
  $P \vdash \Junl\sigma$, we have that $P \vdash \Junl{(\Eta',y:\sigma)}$, and we can construct a
  new derivation of $P \mid \Eta, y: \sigma \vdash x: \tau$ by \strule{Var}. The remaining cases are
  straightforward by \theIH; in those cases where the context is split, the binding $y:\sigma$ can
  be included in either context (or, in fact, both).
\end{proof}

\begin{proof}[Proof of \thmref{completeness-s}]
By induction on the derivation of $P \mid \Eta \vdash M: \sigma$.
\begin{itemize}
\item \textit{Case \trule{var}.} We have $\Eta = \{ x:\sigma \}$; let $\sigma = (\forall
  \vec{t}.Q \then \tau).$ Pick fresh type variables $\vec{u}$; we have ($\subv{u}{t}Q \then
      \subv{u}{t}\tau) \lessgeneral \sigma$, and so $P, \subv u t Q \mid \Eta \vdash
      x: \subv u t \tau$ by \strule{var}.  As the $\vec u$ are fresh,
      \begin{align*}
        \sigma' &= Gen(\Gamma,P,[\vec u/\vec t]Q \then [\vec u/\vec t]\tau) \\
        &= \forall \vec u. (P,\subv{u}{t}Q) \then \subv{u}{t}\tau
      \end{align*}
      and $(P \mid \sigma) \lessgeneral \sigma'$.
\item \textit{Case \trule{\C}.} Because $(x:\sigma)$ is the only newly duplicated binding in the
  subderivation, and $P \vdash \Junl\sigma$, \theIH gives the required derivation of $Q \mid
  \Gamma,x:\sigma \vdash M: \tau$ such that $(P \mid \sigma) \lessgeneral (\emptyset \mid
  Gen(\Gamma,x:\sigma; Q \then \tau)).$
\item \textit{Case \trule{\W}} follows from \lemref{weakening} and \theIH.
\item \textit{Case \trule{\I\to}.} By \theIH and \lemref{more-specific}, we have a derivation of
  $Q \mid \Gamma \vdashS M: \upsilon$.  It is immediately apparent that if
  $P \vdash \Eta \moreunlimited \phi$ then $P \vdash \Gamma \moreunlimited \phi$, and so we can
  construct the desired derivation by \strule{\I\to}.
\item \textit{Case \trule{\E\to}.} By \theIH and \lemref{more-specific}, we have dervations $P \mid
  \Gamma \vdashS M: \phi \tau \upsilon$ and $P \mid \Gamma' \vdash N: \tau$.  Finally, for any
  $x:\sigma$ in both $\Gamma$ and $\Gamma'$, we have that $P \vdash \Junl\sigma$, so we can suitably
  partition $\Gamma$ and $\Gamma'$ and apply \strule{\E\to}.
\item \textit{Cases \trule{\I\oplus$_i$}, \trule{\E\oplus}, \trule{make}, and \trule{break}} follow
  from similar arguments to those for \trule{\I\to} and \trule{\E\to}.
\item \textit{Case \trule{\I\then}.} From \theIH, we have $Q \mid \Gamma \vdashS
  M: \tau$ such that, letting $\sigma = Gen(\Gamma,Q \then \tau)$, $(\pi,P \mid \rho) \lessgeneral
  \sigma$.  As $(P \mid \pi \then \rho) \lessgeneral (\pi, P \mid \rho)$, we also have $(P
  \mid \pi \then \rho) \lessgeneral \sigma$.
\item \textit{Case \trule{\E\then}.} From \theIH, we have $Q \mid \Gamma \vdashS
  M: \tau$ such that, letting $\sigma = Gen(\Gamma,Q \then \tau)$, $(P \mid \pi \then \rho) \lessgeneral
  \sigma$.  Since $P \entails \pi$, we have $(P \mid \rho) \lessgeneral (P \mid \pi \then \rho)$, and so $(P
  \mid \rho) \lessgeneral \sigma$.
\item \textit{Case \trule{\I\forall}.} From \theIH, we have $Q \mid \Gamma
  \vdashS M: \tau$ such that, letting $\sigma' = Gen(\Gamma,Q \then \tau)$, $(P \mid \pi \then \rho)
  \lessgeneral \sigma'$.
\item \textit{Case \trule{\E\forall}.} From \theIH, we have $Q \mid \Gamma \vdashS M: \tau$ such
  that, letting $\sigma' = Gen(\Gamma,Q \then \tau)$, $(P \mid \pi \then \rho) \lessgeneral \sigma$.
  As $(P \mid [\tau/t]\sigma) \lessgeneral (P \mid \sigma)$, $(P \mid [\tau/t]\sigma) \lessgeneral \sigma'$.
\item \textit{Case \trule{let}.}  From \theIH, we have
  $Q \mid \Gamma,x:\forall \tau. \sigma \vdashS M: \tau$ such that, letting
  $\sigma' = Gen(\Gamma,Q \then \tau)$, $(P \mid \forall t. \sigma) \lessgeneral \sigma'$.  Thus, we
  conclude that $\Gamma,x:\sigma' \lessgeneral \Gamma,x:\forall t. \sigma$ and, applying
  \lemref{more-general-context}, \theIH, and \lemref{more-specific}, we have a derivation of
  $Q' \mid \Gamma,x:\sigma' \vdash N: \upsilon$.  Finally, we apply \strule{let} to conclude
  $Q,Q' \mid \Gamma \vdashS \mlet x M N: \upsilon$. \qedhere
\end{itemize}
\end{proof}

\subsection{Type Inference}

We begin with the soundness of the inference algorithm.

\begin{fake}{Theorem \ref{thm:soundness-m} \textnormal{(Soundness of $\M$)}.}
  If $\Mapp S X \Gamma M \tau = \Mres P {S'} \Sigma$, then
  $S'\,P \mid S'\,(\restrict \Gamma \Sigma) \vdash M: S'\,\tau$.
\end{fake}

\noindent
The unusual aspect of the proof is the introduction of improving substitutions during type
inference.  We must show that their introduction does not compromise the soundness of the
corresponding derivations.  We begin by giving a more formal characterization of the possible
improving substitutions for Quill constraints.  We will restrict our attention to simplified
constraints.

\begin{defn}
  A constraint is simple in (the type variable) $t$ if the constraint is of the form $\Unl t$,
  $\Fun t$, or $\tau \geq t$.  A constraint is simple if it is simple in some $t$.  An entailment
  $P \entails Q$ is non-trivial if the only uses of the assumption rule are for simple constraints.
\end{defn}

\noindent
Relying on simple constraints does not limit the expressiveness of the type system.

\begin{lemma}
  Suppose that there is a non-trivial entailment $P \entails Q$; then there is some $Q'$ such that
  $Q'$ is simple and $P \entails Q' \entails Q$.  We call $Q'$ the simplification of $Q$.
\end{lemma}

\begin{proof}
  By induction on the derivation of $P \entails Q$; the cases are individually straightforward.  As
  an illustrative case, if $\Unl{(\tau_1 \oplus \tau_2)} \in Q$, then the simplification of
  $\Set{\Unl\tau_1,\Unl\tau_2}$ is a subset of $Q'$.
\end{proof}

\noindent
We can now define improving substitutions for simple constraints.

\begin{defn}
  Suppose that $Q$ is simple, and $X$ is some set of type variables.  We define a substitution $S$,
  called the improving substitution for $X$ in $Q$, as follows:
  \begin{itemize}
  \item If $\Unl t \in Q$, then $S\,t = {\uto}$; and,
  \item If $\Unl t \not\in Q$, then $S\,t = {\lto}$.
  \end{itemize}
\end{defn}

\begin{lemma}\label{thm:improvement}
  Suppose that $Q \mid \Gamma \vdashS M : \tau$, $X = ftv(Q) \setminus ftv(\Gamma,\tau)$,
  $Q' \then Q$ is a simplification of $Q$, and $S$ is an improving substitution for $X$ in $Q'$.
  Then $S\,Q \mid \Gamma \vdashS M : \tau$.
\end{lemma}

\begin{proof}
  By induction on the derivation of $Q \mid \Gamma \vdashS M : \tau$.  The key observation is that,
  if some variable $f$ is bound in the improving substitution then it must be used in rules
  \strule{\I\to},\strule{\E\to}, and then $S\,f$ is a suitable function type for any of its uses.
\end{proof}

Next, we account for routine manipulations of syntax-directed typing derivations.  Strengthening the
assumed context preserves typing.

 \begin{lemma}\label{thm:expand-context}
  If $P \mid \Delta \vdashS M: \tau$, and $Q \then P$, then $Q \mid \Delta \vdashS M: \tau$.
\end{lemma}

\noindent
The syntax-directed typing system is closed under substitution:

\begin{lemma}\label{thm:closed-substitution-s}
  If $P \mid \Delta \vdash M: \tau$, then $S\,P \mid S\,\Delta \vdash M: S\,\tau$.
\end{lemma}

\noindent
We can add bindings to the environment of a typing derivation, so long as they are unlimited .

\begin{lemma}\label{thm:weakening-s}
  If $P \mid \Delta \vdashS M: \tau$, $y$ is not free in $M$, and $P \vdash \Junl\sigma$, then $P
  \mid \Delta,y:\sigma \vdashS M: \tau$.
\end{lemma}

\noindent
The proof is by induction over the derivation of $P \mid \Delta \vdashS M: \tau$, similarly to
that for \lemref{weakening}.  Finally, we can show that type inference respects the rigid type
variables (assuming the same of the unification algorithm):

\begin{lemma}\label{thm:rigid-variables}
  If $\Mapp S X \Gamma M \tau = \Mres P R \Sigma$, and $\restrict S X = id$, then
  $\restrict R X = id$.
\end{lemma}

\noindent
The proof is by induction on $M$.

Finally, we are prepared to show the soundness of the inference algorithm.

\begin{proof}[Proof of \thmref{soundness-m}]
By induction on the structure of $M$:
\begin{itemize}
\item \textit{Case $x$.} We have that $(x: \forall \vec t. P \then \tau) \in \Gamma$ and
  $\Sigma = \Set{x}$. We see immediately that
  \[
    ([u_i/t_i]P \then [u_i/t_i] \tau) \lessgeneral (\forall \vec t. P \then \tau).
  \]
  So, we can apply \strule{var} to construct a derivation of
  \[
    [u_i/t_i]P \mid \Set{x: \forall \vec t. P \then \tau} \vdashS x: [u_i/t_i]\tau.
  \]
\item \textit{Case $\lambda x.M$.}  We have that
  \[
    \Mres{P}{S'}{\Sigma} = \Mapp {\Unif X \tau {u_1\,u_2\,u_3} \circ S} X {\Gamma,x:u_2} M {u_3}
  \]
  and so, by \theIH,
  \[
    S' \, P \mid S'\, (\restrict{(\Gamma, x : u_2)}{\Sigma}) \vdashS M: S' \, u_3.
  \]
  Let
  $Q = \Set{\Fun u_1} \cup Leq(u_1,\RestrictEnv{\Gamma}{\Sigma}) \cup Weaken(x,u_2,\Sigma) \cup P$
  and let $\Sigma' = \Sigma \setminus \Set{x}$.  By Lemmas~\ref{thm:expand-context} and
  \ref{thm:weakening-s} we can construct a derivation of
  \[
    S' \, Q \mid S' \, (\restrict{(\Gamma, x: u_2)}{\Sigma'}) \vdashS M : S' \, u_3.
  \]
  The desired result is then immediate by \strule{\I\to}.
\item \textit{Case $M \, N$.} We have that
  \begin{align*}
    \Mres P R \Sigma &= \Mapp S X \Gamma M {u_1\,u_2\,\tau} \\
    \Mres {P'} {R'} {\Sigma'} &= \Mapp R X \Gamma N {u_2}
  \end{align*}
  Let
  $Q = Un(\RestrictEnv{\Gamma}{\Sigma'}) \cup Weaken(x,u_1,\Sigma_N) \cup Weaken(y,u_2,\Sigma_{N'})$
  and $Q' = P \cup P' \cup Q$.  Let $\Gamma' = \restrict{\Gamma}{\Sigma \cap \Sigma'}$,
  $\Delta = \restrict{\Gamma}{\Sigma \setminus \Sigma'}$ and
  $\Delta' = \restrict{\Gamma}{\Sigma' \setminus \Sigma}$ (and note that these partition
  $\restrict{\Gamma}{\Sigma \cup \Sigma'}$).  By Lemmas~\ref{thm:expand-context}
  and~\ref{thm:closed-substitution-s} and \theIH, we have derivations of the following:
  \begin{gather*}
    R'\,Q' \mid R'\,\Gamma', R'\,\Delta \vdashS M : R' (u_1 \, u_2 \, \tau) \\
    R'\,Q' \mid R'\,\Gamma', R'\,\Delta' \vdashS N : R' u_2.
  \end{gather*}
  Finally, note that by construction $Q' \then \Junl{\Gamma'}$ and $Q' \then \Fun{u_1}$, so the desired
  result follows from an application of \strule{\E\to}.
\item \textit{Case $\ini \, M$} is immediate by \theIH.
\item \textit{Case $\case{M}{x}{N}{y}{N'}$} follows a very similar argument to that for application.
\item \textit{Case $\mlet{x}{M}{N}$.} We have that
  \begin{gather*}
    \Mres P R \Sigma = \Mapp S X \Gamma M {u_1} \\
    \Mres {P'}{R'}{\Sigma'} = \Mapp R X {\Gamma,x : \sigma} N \tau
  \end{gather*}
  where $\sigma = GenI(R\,\Gamma, R\,(P \then u_1))$.  Let $T$ improve
  $ftv(P) \setminus ftv(\Gamma,R\,u_1)$ in $P$ and.  Then there is a partition of $\Gamma$ into
  $\Gamma_M$, $\Gamma_N$, and $\Delta$ such that, by Lemmas~\ref{thm:improvement},
  \ref{thm:closed-substitution-s} and \theIH we have
  \begin{gather*}
    (T \circ R') \, P \mid R' \, (\Gamma_M,\Delta) \vdashS M : R\,u_1 \\
    R'\,P' \mid R'\,(\Gamma_N,\Delta,x:\sigma) \vdashS N : R'\,\tau
  \end{gather*}
  and the result follows by an application of \strule{Let}.
\item \textit{Case $K \, M$.} We have that
  $K : \forall \vec{t_1}. (\forall \vec{t_2}. \exists \vec{t_3}. Q' \then \phi') \uto \phi$; let
  $(\forall \vec{t_2}. \exists \vec{t_3}. (Q \then \upsilon') \uto \upsilon$ be an instance of that
  type such that $U = \Unif X \upsilon \tau$ does not fail and $U \upsilon = \upsilon$.  (If such an
  instance did not exist, type inference would fail.)  We have that
  \[
    \Mres P R \Sigma = \Mapp {U \circ S} {X \cup \vec{t_2}} {\Gamma} {M} {[\vec{u_3}/\vec{t_3}]\upsilon'}.
  \]
  By \theIH, we have that
  \[
    R\,P \mid R\,(\restrict{\Gamma}{\sigma}) \vdashS M : R\,([\vec{u_3}/\vec{t_3}]\upsilon').
  \]
  The side condition $P \cup [\vec{u_3}/\vec{t_3}] Q \then [\vec{u_3}/\vec{t_3}] Q$ holds trivially,
  and $\vec{t_2} \not\in ftv(P, R\,\Gamma)$ is assured by \lemref{rigid-variables} and the side conditions in
  $\M$.
\item \textit{Case $\mlet{K\,x}{M}{N}$.} We have that
  \[
    \Mres {P_M} R {\Sigma_M} = \Mapp S X \Gamma M \upsilon.
  \]
  As in the previous case, let
  $(\forall \vec{t_2}. \exists \vec{t_3}. (Q \then \upsilon') \uto \upsilon$ be an instance of the
  type of $K$. (If there is not such an instance, type inference fails.)  By \theIH and
  Lemmas~\ref{thm:weakening-s} and~\ref{thm:closed-substitution-s}, we have that there is a
  partition of $\restrict \Gamma {\Sigma_M \cup \Sigma_N}$ into $\Gamma_M, \Gamma_N, \Delta$ such that
  \begin{gather*}
    R'\,P \mid R' (\Gamma_M, \Delta) \vdashS M : \upsilon \\
    R'\,(P \cup [\vec{u_2}/\vec{t_2}] Q) \mid R' \, (\Gamma_N, \Delta') \vdashS N : \tau
  \end{gather*}
  and the side condition is assured by \lemref{rigid-variables}. \qedhere
\end{itemize}
\end{proof}

Completeness of the inference algorithm is relatively straightforward.  We begin with a lemma
characterizing the effect of the input substitution.

\begin{lemma}\label{thm:infer-more-specific}
  If $\Mapp S X \Gamma M \tau = \Mres P {S'} \Sigma$, then
  $\Mapp {id} X \Gamma M \tau = \Mres {P'} {S''} \Sigma$ where $P \entails S\,P'$ and
  $S' = S \circ S''$.
\end{lemma}

\begin{proof}
  By induction on the structure of $M$.
\end{proof}

We can now show that the algorithm is complete.

\begin{fake}{Theorem \ref{thm:completeness-m} \textnormal{(Completeness of $\M$)}.}
  If $S$ is a substitution such that $P \mid S\,\Gamma \vdashS M : S\,\tau$, and
  $\restrict S X = id$, then $\Mapp {id} X \Gamma M \tau = \Mres Q {S'} \Sigma$ such that
  $(P \then S\,\tau) \lessgeneral GenI(S'\,\Gamma, S'\,Q \then S'\,\tau)$.
\end{fake}

\begin{proof}
  We show the result for $\Mapp S X \Gamma M \tau$ by induction on the height of the derivation of
  $P \mid S\,\Gamma \vdashS M: S\,\tau$, and then apply \lemref{infer-more-specific} to show the
  theorem.
  \begin{itemize}
  \item \textit{Case \strule{var}.}  We have that $(x:\sigma) \in \Gamma$ such that
    $(P \then S\,\tau) \lessgeneral \sigma$.  Let $\sigma = \forall \seq t. Q \then \upsilon$.
    Then, we have that
    \[
      \Mapp S X \Gamma x \tau = \Mres {([\vec u/\vec t]\,Q)} {U \circ S} {\Set{x}}
    \]
    where $U = \Unif X {[\seq u/\seq t]\,\upsilon} {S\,\tau}$.  By assumption $U \, [\vec u/\vec t]
    \upsilon = S\,\tau$ and $P \then U\,([\vec u/\vec t]\,Q)$, so $P \then S\,\tau \lessgeneral
    GenI(S\,\Gamma, S\,Q \then S\, \tau)$.
  \item \textit{Case \strule{\I\to}.}  We have a derivation concluding
    $P \mid \Gamma \vdashS \lambda x. M : \phi\,\tau\,\upsilon$ such that $P \then \Fun \phi$ and
    $P \then \Gamma \moreunlimited \phi$.  It is immediate that
    $\Unif X {\phi\,\tau\,\upsilon} {u_1\,u_2\,u_3}$ will give the unifier
    $[\phi/u_1,\tau/u_2,\upsilon/u_3]$, so we will assume that unifier for the remainder of this
    case.  By \theIH, we have that
    \[
      \Mapp S X {\Gamma, x:\tau} M {\upsilon} = \Mres {Q'} R \Sigma
    \]
    (the role of the generalization is unimportant). Therefore, we see that
    \[
      \Mapp S X \Gamma {\lambda x. M} {\phi \tau \upsilon} = \Mres {Q} R {\Sigma \setminus x}
    \]
    where $Q = Q' \cup \Set{\Fun \phi} \cup Leq(\phi, \Gamma) \cup Weaken(x,\tau,\Sigma)$.  That $P
    \then Q$ follows from the assumption that $P \then Q'$, the side conditions of the initial
    derivation, and the side conditions of any uses of \strule{Var} in the initial derivation.
  \item \textit{Case \strule{\E\to}.}  We have a derivation concluding
    $P \mid S\,\Gamma \vdashS M\,N : \upsilon$ (where $\upsilon = S\,\tau$).  By \theIH, we can
    conclude that
   \begin{gather*}
     \Mapp S X \Gamma M {\phi \, \tau' \, \upsilon} = \Mres Q S \Sigma \\
     \Mapp S X \Gamma N {\tau'} = \Mres {Q'} S {\Sigma'}
   \end{gather*}
   where $P \then Q$ and $P \then Q'$. So, we have that
   \[
     \Mapp S X \Gamma {M \, N} \upsilon = \Mres {Q''} S {\Sigma \cup \Sigma'}
   \]
   where $Q'' = Q \cup Q' \cup \Set{\Fun \phi} \cup Un(\restrict \Gamma {\Sigma \cap \Sigma'})$.
   Finally, the side conditions of the initial derivation ensure that $P \then \Fun \phi$ and
   $P \then Un(\restrict \Gamma {\Sigma \cap \Sigma'})$ and so $P \then Q''$.
 \item \textit{Cases \strule{\I\oplus$_i$} and \strule{\E\oplus}} follow from similar arguments to
   those for \strule{\I\to} and \strule{\E\to}.
 \item \textit{Case \strule{let}.}  We have a derivation concluding
   $P \mid S\,\Gamma \vdashS \mlet{x}{M}{N} : \upsilon$ (where $\upsilon = S\,\tau$, and $\Gamma$ is
   partitioned into $\Gamma_M$, $\Gamma_N$, and $\Delta$).  From the subderivation of
   $Q \mid S\,(\Gamma_M,\Delta) \vdashS M : \tau'$ and \theIH conclude that
   \[
     \Mapp S X \Gamma M {u_1} = \Mres Q' {S'} \Sigma
   \]
   such that $Gen(S\,\Gamma, S\,(Q \then \tau')) \lessgeneral GenI(S\,\Gamma, S'\,(Q' \then u_1))$.
   Then, from the subderivation of
   $P \mid S\,(\Gamma_N,\Delta, x:Gen(S\,\Gamma, S\,(Q \then \tau')) \vdash N: \upsilon$,
   \lemref{more-general-context}, and \theIH, we conclude that
   \[
     \Mapp {S'} X {\Gamma, x: GenI(S\,\Gamma, S\,(Q' \then u_1))} N \upsilon = \Mres P' {S'} \Sigma'
   \]
   where the side conditions on the initial derivation are sufficient to ensure that
   $P \then Un(\restrict \Gamma {\Sigma \cap \Sigma'})$ while the side conditions on uses of
   \strule{var} assure that $P \then Weaken(x, GenI(S\,\Gamma, S'\,(Q' \then u_1)), \Sigma')$.
 \item \textit{Cases \strule{make} and \strule{break}} follow from \theIH, with similar arguments as
   for cases \strule{\I\to} and \strule{\E\to}; the eigenvariable conditions in \strule{make} and
   \strule{break} are sufficient to ensure that the disjointness conditions in $\M$ hold.
 \end{itemize}
\end{proof}

Finally, we can build on the soundness and completeness of the type inference algorithm and the
syntax-directed type system to give an effective proof of principal types.

\begin{fake}{Theorem \ref{thm:principal} \textnormal{(Principal Types)}.}
  If $P_0 \mid \Eta \vdash M: \sigma_0$ and $P_1 \mid \Eta \vdash M: \sigma_1$ then there is
  some $\sigma$ such that $\emptyset \mid \Eta \vdash M: \sigma$ and $(P_0 \mid \sigma_0) \lessgeneral
  \sigma, (P_1 \mid \sigma_1) \lessgeneral \sigma.$
\end{fake}

\begin{proof}
  Suppose that $P_0 \mid \Eta \vdash M: \sigma_0$.  From \thmref{soundness-s} we have there
  there are some $Q_0$ and $\tau_0$ such that $Q_0 \mid \Gamma \vdashS M: \tau_0$ and $(P \mid
  \sigma_0) \lessgeneral Gen(\Gamma, Q_0 \then \tau_0)$.  Similarly, from $P_1 \mid \Eta \vdash
  M: \sigma_1$, we have $Q_1 \mid \Gamma \vdashS M: \tau_1$ such that $(P_1 \mid \sigma_1)
  \lessgeneral Gen(\Gamma, Q_1 \then \tau_1)$.  From \thmref{completeness-s}, we have that
  $W(\Gamma,M) = Q;S\Delta;\upsilon$ such that, writing $\sigma = Gen(\Gamma, Q \then \upsilon)$,
  $Gen(\Gamma,Q_0 \then \tau_0) \lessgeneral \sigma$ and $Gen(\Gamma,Q_1 \then \tau_1)
  \lessgeneral \sigma$.  Finally, by transitivity, we have that $(P_0 \mid \sigma_0) \lessgeneral
  \sigma$ and $(P_1 \mid \sigma_1) \lessgeneral \sigma$.
\end{proof}

\subsection{Conservativity of Typing}

We now show that Quill is a conservative extension of Jones's core functional calculus OML.  We give
the syntax and typing rules of OML in \figref{oml-terms-typing}.  We overload the meta-variables of
Quill to play similar roles in the definition of OML; the meaning of individual meta-variables will
be apparent from context.  Our presentation of OML differs from Jones's~\cite{Jones94} in two
respects.  First, we associate Jones's function type ($\tau \to \upsilon$) with our unrestricted
function type ($\tau \uto \upsilon$).  This preserves the meaning of OML terms while avoiding the
need to introduce new polymorphism in the interpretation of OML types as \lang types.  It may seem
restrictive, but the principal types theorem for \lang (\thmref{principal}) assures that we can find
more general types for terms if they exist.  Second, we treat sums explicitly, whereas Jones leaves
them implicit (or treated by Church encoding).  This corresponds to the need to introduce one
additive type in \lang; the typing rules we give for for $\tau \oplus \upsilon$ in OML are exactly
those that would arise by encoding.

\begin{figure}
\[\begin{array}{ll@{\hspace{3mm}}ll}
  \text{Term variable} & x,y \in Var & \text{Type variables} & t,u \in TVar \\
  \text{Environments} & \Gamma \\
  \text{Type constructors} & \multicolumn{3}{l}{T^\kappa \in \T^\kappa \text{ where $\Set{\oplus,\uto} \subseteq \T^{\star \to \star \to \star}$}}
\end{array}\]
\begin{longsyntax}
  \text{Kinds} & \tcr{\kappa} & ::= & \star \mid \kappa \to \kappa \\
  \text{Types} & \tcr{\tau^\kappa} & ::= & t \mid T^\kappa \mid \tau^{\kappa' \to \kappa}\,\tau^{\kappa'} \\
  \tcl{\text{Predicates}} & \pi & ::= & \dots \\
  \tcl{\text{Qualified types}} & \rho & ::= & \tau^\star \mid \pi \then \rho \\
  \tcl{\text{Type schemes}} & \sigma & ::= & \rho \mid \forall t. \sigma \\
  \text{Expressions} & \tcr{M,N} & ::= & x \mid \lambda x. M \mid M\,N \mid \inl{M} \mid \inr{N} \\
   &&& \mid & \case{M}{x}{N}{y}{N'} \\
   &&& \mid & \mlet{x}{M}{N}
\end{longsyntax}
\begin{gather*}
\infbox{\irule{(x : \sigma) \in \Gamma}
              {(P \then \tau) \lessgeneral \sigma};
              {P \mid \Gamma \vdashSO x : \tau}}
\quad
\infbox{\irule{P \mid \Gamma, x : \tau \vdashSO M : \upsilon};
              {P \mid \Gamma \vdashSO \lambda x. M : \tau \uto \upsilon}}
\\
\infbox{\irule{\begin{array}{c}
                 {P \mid \Gamma \vdashSO M : \tau \uto \upsilon}
                 \\
                 {P \mid \Gamma \vdashSO N : \tau}
               \end{array}};
              {P \mid \Gamma \vdashSO M\,N : \upsilon}}
\quad
\infbox{\irule{P \mid \Gamma \vdashSO M : \tau_i};
              {P \mid \Gamma \vdashSO \ini M : \tau_1 \oplus \tau_2}}
\\
\infbox{\irule{\begin{array}{c}
                 {P \mid \Gamma \vdashSO M : \tau_1 \oplus \tau_2}
                 \hspace\infskip
                 {P \mid \Gamma, x : \tau_1 \vdashSO N : \upsilon}
                 \\
                 {P \mid \Gamma, y : \tau_2 \vdashSO N' : \upsilon}
               \end{array}};
              {P \mid \Gamma \vdashSO \case M x N y {N'} : \upsilon}}
\\
\infbox{\irule{Q \mid \Gamma \vdashSO M : \tau}
              {P \mid \Gamma, x : Gen(\Gamma, Q \then \tau) \vdashSO N : \upsilon};
              {P \mid \Gamma \vdashSO \mlet x M N : \upsilon}}
\end{gather*}
\caption{Terms and typing of OML.}\label{fig:oml-terms-typing}
\end{figure}

We can now show that \lang is a conservative extension of OML, by showing that any syntax-directed
typing of a term in OML corresponds to a syntax-directed typing of the same term in \lang.

\begin{fake}{Theorem \ref{thm:conservative}.}
  If $P \mid \Gamma \vdashSO M : \tau$, then there is some $Q$ such that $Q \mid \Gamma \vdashS
  M : \tau$, and $Q \entails P$.
\end{fake}

\begin{proof}
  The proof is by induction on the structure of the derivation; the cases are all immediate by
  \theIH and \lemref{expand-context}.
\end{proof}

\subsection{Semantics}

We begin by giving the ``predictable'' definition of the subexpressions $SExp(M)$ of an expression
$M$, as follows.

\begin{align*}
  SExp(x) &= \emptyset \\
  SExp(\lambda x.M) &= Exp(M) \\
  SExp(M\,N) &= Exp(M) \cup Exp(N) \\
  SExp\hspace{-3px}\left(\begin{array}{@{}l@{}}\mkwd{case}\,x\,\mkwd{of}\\\{\inl x \mapsto N; \inr y \mapsto N'\}\end{array}\right) &= Exp(M) \cup Exp(N) \cup Exp(N') \\
  SExp(\ini M) &= Exp(M) \\
  SExp(\mlet{x}{M}{N}) &= Exp(M) \cup Exp(N) \\
  SExp(K\,M) &= Exp(M) \\
  SExp(\mlet{K\,x}{M}{N}) &= Exp(M) \cup Exp(N)
\end{align*}

\begin{fake}{Theorem \ref{thm:safety} \textnormal{(Type safety)}.}
  Let $M$ be a closed term such that $\vdash M: \forall t. P \then \tau$ and $\red M V I E$.
  \begin{enumerate}
  \item $P \mid \emptyset \vdash V : \tau$.
  \item Let $E' = E \cup Val(V)$, and let $D = I \setminus E'$ (the values discarded during
    evaluation) and $C = E' \setminus I$ (the values copied during evaluation). Then,
    $W \in D \cup C$ only if $W \not\in \LinV_P$.
  \end{enumerate}
\end{fake}

\begin{proof}
  The proof of (1) is a straightforward induction on the height of the derivation of $\red M V I E$.
  The proof of (2) is done similarly, by cases on the reducing term.  We show an illustrative case;
  the remaining cases are similar.
  \begin{itemize}
  \item \textit{Case $M\,N$.}  We have that $\red M {\lami j x {M'}} I E$, $\red N V {I'} {E'}$, and
    $\red {[V/x]M'} {W} {I''} {E''}$.  Let
    \begin{align*}
      I_0 &= I \cup I' \cup I''; \\
      E_0 &= E '\cup E' \cup E'' \cup \lami j x M'; \text{and,} \\
      E_0' &= E_0 \cup Val(W).
    \end{align*}
    W.L.O.G., suppose that $W' \in I_0 \setminus E_0'$ (i.e., it is discarded during evaluation).
    If this happens during the reduction of $M$ or $N$ then the case holds by \theIH.
    Alternatively, it may happen during the reduction of $[V/x]M'$.  But then the result holds from
    the well-typing of $M$.  For example, suppose that the variable $x$ (of type $\tau$) does not
    appear in $M'$, and so $V$ itself is discarded.  Then for $M$ to be well-typed, it must be the
    case that $P \then \Unl\tau$, and so $V \not\in \LinV_P$. \qedhere
  \end{itemize}
\end{proof}

\raggedend

\end{document}